%% file: main.tex
\documentclass[runningheads]{llncs}
\usepackage[T1]{fontenc}
\usepackage{graphicx}
\usepackage{array}
\usepackage{hyperref}
\usepackage{color}

\usepackage{amsmath}
\usepackage{amssymb}
\usepackage{mathtools}
\usepackage{algorithm}
\usepackage{algorithmic}
\usepackage[switch]{lineno}
\usepackage{subcaption}
\usepackage{csquotes}

\usepackage{cite} 

\usepackage{pgfplots}
\usepackage{mathtools}

\newcommand{\citet}[1]{\citeauthor{#1}~\shortcite{#1}}
\newcommand{\supertiny}[1]{{\fontsize{4.5pt}{4.5pt}\selectfont {#1}}}

\newcommand\blfootnote[1]{%
\begin{NoHyper}
  \begingroup
  \renewcommand\thefootnote{}\footnote{#1}%
  \addtocounter{footnote}{-1}%
  \endgroup
  \end{NoHyper}
}

\DeclareMathOperator*{\argmax}{arg\,max}
\DeclareMathOperator*{\argmin}{arg\,min}

\usepackage{style}

\allowdisplaybreaks

\begin{document}
\title{Contested Logistics: A Game-Theoretic Approach}
%
%
\author{
Jakub \v{C}ern\'{y}$^{\dagger,}$\inst{1} \and
Chun Kai Ling\inst{1} \and
Darshan Chakrabarti\inst{1} \and 
Jingwen Zhang\inst{1} \and
Gabriele Farina\inst{2} \and
Christian Kroer\inst{1} \and
Garud Iyengar\inst{1}
}
\authorrunning{J. \v{C}ern\'{y} et al.}
%
\institute{Columbia University, New York, NY 10027, USA, \email{\{jakub.cerny, chunkai.ling, darshan.chakrabarti, jz3093, christian.kroer, gi10\}@columbia.edu}
\and
MIT, Cambridge, MA 02139, USA,
\email{gfarina@mit.edu}}
\maketitle              
\begin{abstract}
We introduce Contested Logistics Games, a variant of logistics problems that account for the presence of an adversary that can disrupt the movement of goods in selected areas. We model this as a large two-player zero-sum one-shot game played on a graph representation of the physical world, with the optimal logistics plans described by the (possibly randomized) Nash equilibria of this game.
Our logistics model is fairly sophisticated, and is able to handle multiple modes of transport and goods, accounting for possible storage of goods in warehouses, as well as Leontief utilities based on demand satisfied. We prove computational hardness results related to equilibrium finding and propose a practical double-oracle solver based on solving a series of best-response mixed-integer linear programs. We experiment on both synthetic and real-world maps, demonstrating that our proposed method scales to reasonably large games. We also demonstrate the importance of explicitly modeling the capabilities of the adversary via ablation studies and comparisons with a naive logistics plan based on heuristics.
\blfootnote{Equal contribution between J. \v{C}ern\'{y} and C. K. Ling. $^{\dagger}$Corresponding author.}
\keywords{Logistics  \and Game theory \and Equilibrium computation.}
\end{abstract}

\input{sections/intro}
\input{sections/prelims}
\input{sections/model}
\input{sections/experiments}
\input{sections/conclusion}

\begin{credits}
\subsubsection{\ackname} 
This research was supported by the Office of Naval Research award N00014-23-1-2374. Christian Kroer was additionally supported by the Office of Naval Research award N00014-22-1-2530, and the National Science Foundation awards IIS-2147361 and IIS-2238960.
\end{credits}
%
%
%
\bibliographystyle{splncs04}
\bibliography{library}
\end{document}

%% file: sections/intro.tex
\section{Introduction}

Logistics is a multi-million dollar business with applications in numerous real-world domains. In this paper, we study a variant we call Contested Logistics (CL). CL features two players, customarily identified with the names \textit{Blue} and \textit{Red}. Blue is the logistics player, while Red is an interdiction player seeking to reduce Blue's utility. CL captures the strategic interaction between Red and Blue as a two-player zero-sum one-shot game. A solution to the game is identified by the \textit{Nash equilibrium} (NE) solution concept.

CL is motivated by military considerations, where logistics may be disrupted by an adversary, and robustness considerations, where logistics may be disrupted by acts of God, unforeseen failures, political instability, or other factors. Attacks on supply lines have been extensively documented in real military settings, and are often viewed as more effective than direct kinetic confrontation \cite{mcmahon2017maritime}. Similarly, geopolitical powers such as the US, China, and the EU seek to diversify their supply chains with the intention of being robust against a possible outbreak of hostilities \cite{cowen2010geography}. Likewise, the recent Evergreen Suez canal blockage incident is a painful reminder of the potential risks of having a single-point of failure \cite{lee2021suez}.

While the presence of adversaries in logistics is not new \cite{blom2020inventory}, the CL model differs from prior work in that (i) it does not assume a particular behavior of the adversary, instead allowing Red to act in a manner that most hurts Blue, and (ii) we allow for very dramatic attacks by Red, completely destroying routes or segments of railroads, as opposed to relatively tame effects like reducing a route's capacity or introducing small uncertainties in supply or demand.

The inclusion of Red introduces game-theoretic considerations. Since Blue's logistics and Red's interdiction plans are chosen simultaneously, the resultant Nash equilibrium is typically randomized. Additionally, the computation of the equilibrium poses significant challenges. For instance, the number of possible logistics plans is doubly exponential, while the number of interdiction plans grows exponentially with Red's budget. Thus, explicitly specifying the CL problem as a zero-sum bimatrix game is not practical. Our main contributions are as follows: 
\begin{itemize}
\item We formally propose the framework of Contested Logistics (CL) games, a novel variant of logistics planning that accounts for Red's capabilities. Our min-max formulation \textit{explicitly} models Red actively seeking to thwart Blue, via relatively drastic measures compared to prior work. We show that an optimal strategy exists for both players via von Neumann's minimax theorem. 
\item We prove that computing a Nash eq., as well as best responses of Red and Blue, are \NP{}-hard problems. Nonetheless, the best responses of Blue (respectively, Red) to a fixed randomized Red (resp., Blue) strategy can be written compactly as a polynomial-size mixed integer linear program (MILP). 
\item We propose solving CL games via a double oracle method, utilizing our best-response MILPs. We demonstrate scalability via experiments.  
\item  We conduct experiments using \textit{real-world} inspired scenarios, observing the following. (i) Optimal solutions to the CL problem exhibit counter-intuitive behavior, providing insights into what the solution to the CL problem may look like in practice. (ii) A na\"ive, heuristics-based approach for Blue results in a highly exploitable strategy, suggesting that explicitly accounting for Red's capabilities is important. (iii) The cost of overestimating Red's capabilities (\textit{i.e.}, budget) is relatively low, but conversely, underestimating Red's capabilities leads to a drastic decrease in performance, reaffirming the adage that ``it is better to be safe than sorry.''
\end{itemize}

\section{Related Work}

This paper is related to several fields spanning across disciplines. We concentrate on the fields most pertinent to our game-theoretic model. For traditional (non-adversarial) logistics, refer to the established literature~\cite{pfohl1998logistics,ghiani2004introduction,daganzo2005logistics}.

\subsubsection{Logistics and Routing Models}

Logistics in a contested environment, where adversaries actively interfere with supply chain operations, has been explored in various contexts, especially within military logistics~\cite{jaiswal2012military,ausseil2020identifying}.  Many existing models assume a simplified model of Red, who acts blindly or with limited information and follows a fixed (deterministic or stochastic) behavior~\cite{barahona2007inventory,salmeron2009stochastic,hill2009overview}. While bi-level optimization is sometimes incorporated, the solutions typically remain deterministic, limiting their ability to adapt to more dynamic adversaries~\cite{bell2015military}.

Vehicle routing problems involve optimizing routes for vehicles delivering goods or services. The literature on (robust) routing strategies is extensive, but the typical sources of uncertainty in these models are costs, demands, time windows, or customers, rather than adversaries~\cite{ordonez2010robust,agra2013robust,sungur2008robust,lee2012robust}. Models involving adversarial elements face similar challenges as those in the logistics literature. They often assume either simplistic probabilistic models~\cite{blom2020inventory,alotaibi2018unmanned} or bi-level models with a single vehicle, as seen in ambush avoidance or hazardous materials transport literature~\cite{erkut1998modeling,salani2010ambush,list1991modeling}. Alternatively, they provide deterministic solutions, as in routing interdiction problems~\cite{church2004identifying,sadati2020r,bidgoli2018arc}.

\subsubsection{Game-Theoretic Models}

Network interdiction games explore optimal arcs in a network for interdiction purposes, initially studied in \cite{wollmer1964removing} and applied in cybersecurity, cyberphysical security, or supply-chain attacks \cite{wollmer1964removing,washburn1995two,smith2020survey,smith2008algorithms}. These models typically focus on disrupting traversal paths without accounting for the coordination required among multiple connectors, crucial in logistics scenarios.

Security games have seen practical applications, with defenders choosing distributions over targets and attackers selecting targets to attack \cite{jain2013security,pita2008armor,shieh2012protect,an2017stackelberg}. The simplest versions of such games enjoy polynomial-time solvers, even in the general-sum case \cite{kiekintveld2009computing,conitzer2006computing}. Many developments have been made to account for large but structured strategy spaces such as defender target schedules \cite{korzhyk2010complexity} and repeated interactions \cite{fang2015security}. While efficient in many cases, they often simplify strategies and lack modeling depth for logistics movement and coordination.

Another notable class of games are extensive-form games (EFGs), which are played on game trees where players decide actions at each information set \cite{shoham2008multiagent}. Notably these were used to generate superhuman poker AIs~\cite{bowling2015heads,brown2018superhuman,brown2019superhuman}. 
For CL settings, EFGs could be used to model sequential CL problems, though the action spaces would become potentially prohibitively large. In this paper we focus on single-shot CL problems, which are easier to model and solve, but sequential CL problems are an interesting future direction.

%% file: sections/prelims.tex
\section{Contested Logistics}

CL games are played on a directed graph whose nodes correspond to different types of locations---cities, provinces, towns, \emph{et cetera}. There are several types of \textit{packages} (that is, resources) available that may be transported. Some of the nodes are specially designated as \textit{demand}, \textit{supply}, or \textit{warehouse} nodes. At warehouses, packages may be dropped off and stored. To facilitate transportation of packages, there are several \textit{connectors} (for instance, trucks, trains, or planes) which may be used to transport packages between locations; what a connector can carry, its capacity, and where it can traverse, i.e., edges in the graph, are connector-specific. For example, aircraft cannot carry packages that are too heavy, and while trains have a larger capacity than trucks, they are restricted to traversing only railroads.

Given these specifications, the game proceeds as follows. Red chooses a set of edges to \textit{interdict}, subject to budget constraints. Blue then decides what, where, when, and how packages are sent from supply to demand nodes using the connectors available. Blue aims to satisfy as much demand as possible within a specified time horizon. The game is zero-sum, i.e., Red's goal is to minimize demand satisfied. We adopt a two-stage approach for Blue's logistic plans. In the routing phase, Blue selects where each connector should be routed without committing to any loading. The individual routes can (and often are) correlated across connectors, but are chosen \textit{concurrently} with Red's decision of where to interdict. In the loading phase, Blue observes where Red has chosen to interdict, and uses this information to select a suitable load for connectors, \textit{without changing their routes}. Any connector that was interdicted is forbidden from carrying loads after the point of interdiction, but may still be utilized prior to that. This two-phase approach was introduced by~\cite{bertsimas2016power}, and models situations where unlike routing, loading decisions can be changed easily and on-the-fly.
The approach also allows some level of recourse by Blue, while still having a single-shot zero-sum game model, which is preferable from a computational standpoint.

Formally, we represent a CL game as a directed \textit{physical graph} $\mathsf{G} = (\mathsf{V}, \mathsf{E})$. The nodes in $\mathsf{G}$ represent locations in the physical world Blue traverses. The edges $\mathsf{E}$ can be interdicted by Red, affecting Blue's ability to enact their logistics.

\subsection{Blue's Strategy Space}

On the physical graph $\pgraph = (\pvertices, \pedges)$, Blue has a subset of nodes $\warehouses\subseteq\pvertices$ designated as warehouses, where they can store packages that are currently not being moved around. We assume there is at most one warehouse in each node. Each warehouse has an initial (possible zero) supply, given by a non-negative function $\supply:\warehouses\times\packages\to\realp_0$, and a demand for packages, given by a function $\demand:\warehouses\times\packages\to\realp_0$.

Moving the packages is done by a set of connectors $\connectors$. With each connector $c\in\connectors$ there is an associated subset of edges $\pedges_c\subseteq \pedges$ the connector may use to move across the physical graph, and a function $M:\connectors\times\pedges\to\mathbb{Z}^+$ determining how many timesteps does it take the connector to cross an edge. We assume this value is infinite for the edges the connector cannot use. In addition, each connector has a designated initial location given by a function  $\initlocation:\connectors\to \warehouses$, and weight and volume capacities given by functions $W_{\text{max}}:\connectors\to\realp$ and $V_{\text{max}}:\connectors\to\realp$, respectively. Moreover, Blue has a set of package types $\packages$, each with its associated single unit weight $W:\packages\to\realp$, and single unit volume $V:\packages\to\realp$. Finally, all movement happens over a finite number of discrete timesteps $\timesteps = \{0,1,\dots,T\}$.

For a given connector $c \in \connectors$, based on its initial location $\initlocation(c)$, movement speed $M(c,e)$, accessible edges $\pedges_c$, and timesteps $\timesteps$, we unroll the physical graph into individual connector-specific layered directed graphs $\lgraph_c = (\lvertices, \ledges_c)$, where $\lvertices = (\mathsf{V}_t)_{t\in\timesteps}$ is a series of copies of the physical nodes spread across time. For an edge $e\in\ledges$, we denote by $\lvertices^-(e)$ the tail node of $e$ and $\lvertices^+(e)$ the head node. The edges between the individual layers are found by a simple breadth-first search from the connector's initial location. We assume that no connector can cross more than a single edge in the physical graph in one timestep. However, note that the edges can jump layers, in case it takes the connector more than one timestep to cross an edge. Note that this unrolling process \textit{does not} create a exponentially sized tree but a compact layered DAG (see Figure~\ref{fig:unrolled}).

\begin{figure}[t]
    \centering
    \begin{subfigure}[b]{0.3 \linewidth}
    \centering
    \input{figs/physical-graph}   
    \caption{$\pgraph$}
    \label{fig:physical-graph}
    \end{subfigure}
    \begin{subfigure}[b]{0.3 \linewidth}
    \centering
    \input{figs/unrolled-c1}
    \caption{$\mathcal{G}_{c_1}$}
    \label{fig:unrolled-c1}
    \end{subfigure}
    \begin{subfigure}[b]{0.3 \linewidth}
    \centering
    \input{figs/unrolled-c2}
    \caption{$\mathcal{G}_{c_2}$}
    \label{fig:unrolled-c2}
    \end{subfigure}
    \caption{Physical graph $\pgraph$ and layered graphs $\mathcal{G}_{c_1}, \mathcal{G}_{c_2}$ obtained by unrolling $\pgraph$ over 3 steps. 
    Connectors $c_1$ and $c_2$ start at A and B, respectively. $\mathsf{G}$ has a loop at A for $c_1$ only, taking 2 steps to cross. All the other edges can be crossed in a single timestep by either connector. Unreachable nodes are in white.
    }
    \label{fig:unrolled}
\end{figure}
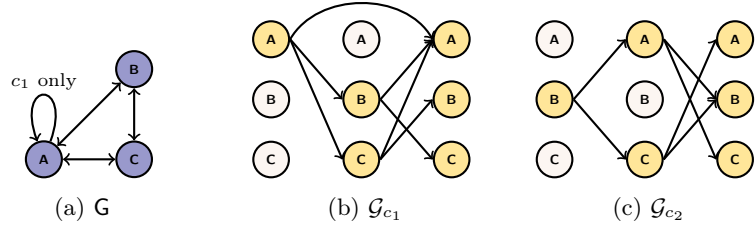

Due to the construction, the layered graph edges in general differ across the connectors, whereas the nodes in the individual layers are the same. Each $v\in\lvertices$ corresponds to some node $\mathsf{v}\in\mathsf{V}$ laying in layer $t$ and we denote this copy of node $\mathsf{v}$ as $v=\mathsf{v}_t$. For an edge $e\in\ledges_c$, we denote the corresponding edge in the physical graph as $\mathsf{E}(e)$. $\mathsf{E}(e)$ is always a singleton. Blue's action space consists of paths in these layered graphs, one per each connector, and can be encoded as solutions to the following feasibility MILP:
\begin{equation}
    \label{blue-flows}
    \tag{F}
    \begin{aligned}
        1 &= \sum_{e\in \ledges_c^-(\initlocation(c)_0)}f_{c}(e)&&\forall c\in \connectors \\
        \sum_{e\in \ledges_c^-(v_t)}f_{c}(e) &= \sum_{e\in \ledges_c^+(v_t)}f_{c}(e)&&\forall c\in \connectors,\;\forall t\in\timesteps,\;\forall v\in\pvertices\\
        f_c(e)&\in\{0,1\}&&\forall c\in\connectors,\;\forall e\in\ledges_c. 
    \end{aligned}
\end{equation}
We call a feasible tuple of connector paths a \textit{logistics plan} and denote it $\blueplan\in\blueplans$, with $\blueplan_c=(e_{c,1},\dots,e_{c,k})$ being a path of a connector $c$. For each logistics plan we have (potentially many) associated feasible package flows, described by the following set of constraints R$(\blueplan)$, starting with the initial supply equation
\numberwithin{equation}{subsection}
\renewcommand{\theequation}{R.\arabic{equation}}
\begin{align}
    S(w,p) &= s_{w_0,p}&&\forall w\in\warehouses,\;\forall p\in\packages,
    \label{eq:R-constrs-first}
\end{align}
then the flow conservation constraints, distinguishing between physical locations that serve or serve not as warehouses
\begin{align}
    l_{c,p}(e_{c,i}) &= l_{c,p}(e_{c,i+1}) && \forall c\in\connectors,\;\forall p\in\packages \notag\\
    & &&\forall e_{c,i}\in\blueplan_c: V^+(e_{c,i})\not\in\warehouses\\
        s_{w_t,p} + \smashoperator{\sum_{c\in\connectors,e\in\blueplan_c: \mathcal{V}^+(e)=w_t}}l_{c,p}(e) &= s_{w_{t+1},p} + \smashoperator{\sum_{c\in\connectors,e\in\blueplan_c: \mathcal{V}^{-}(e)=w_t}}l_{c,p}(e) && \forall t\in\timesteps,\;\forall p\in\packages,\;\forall w\in\warehouses,
\end{align}
and the weight and volume limits of each connector
\begin{align}
    W_{max}(c)&\geq \sum_{p\in\packages} W(p)l_{c,p}(e) &&\forall c\in\connectors,\;\forall e\in \blueplan_c \\
    V_{max}(c)&\geq \sum_{p\in\packages} V(p)l_{c,p}(e) &&\forall c\in\connectors,\;\forall e\in \blueplan_c  \\
    l_{c,p}(e)&\geq 0 && \forall c\in \connectors,\;\forall p\in \packages,\;\forall e\in \blueplan_c\\
    s_{w_t,p}&\geq 0 && \forall t\in \timesteps\cup\{T+1\},\;\forall p\in \packages,\;\forall w\in \warehouses,
    \label{eq:R-constrs-last}
\end{align}
where the $l$ variables encode the package flows, while the $s$ variables record the amount of packages stored in warehouses.

\subsection{Red's Strategy Space}

The strategy space of Red is significantly simpler than Blue's. Our model is similar to the classic network interdiction problems, where Red chooses a subset of edges in the physical graph $\pgraph$ to interdict, given a budget $B \geq 0 $ and a cost function $C:\pedges\to\realp$. Since Red operates on the physical graph instead of any layered graph, we assume an edge is interdicted over the entire game. Red's action space is formed by all feasible solutions of the following MILP:
\begin{equation}
    \label{red-interdiction}
    \tag{Y}
        \bigg\{ y \in \{ 0, 1 \}^{|\pedges|} \: \big| \: B \geq \sum_{e\in \pedges}C(e)y(e) \bigg\}.
\end{equation}
We call a feasible set of interdicted edges an \textit{interdiction plan} and denote it $\redplan\in\redplans$.

\subsection{Utilities}

We assume that Blue aims to maximize the (cumulative) Leontief value at each location with demand $D$ at the final timestep $T$, given a logistics plan $\blueplan$ and feasible package flows $l$ and $s$. The value is defined as
\begin{equation}\label{leontief}
    v(\blueplan,s,l) = \sum_{w\in\warehouses}P(w)\max\left\{\min_{p\in\packages, \demand(w,p)>0}\frac{s_{w_{T},p}}{\demand(w,p)},U(w)\right\},\tag{L}
\end{equation}
where $P(w),U(w)$ are warehouse-specific payoffs per each unit and maximum numbers of units, respectively, of satisfied demand. 

Motivated by the randomized network interdiction problems~\cite{bertsimas2016power}, we make the following two assumptions about the effects of Red's interdiction plan $\redplan$ on Blue's logistics $\blueplan$. (i) Whenever a connector attempts to cross an interdicted physical edge, it is destroyed together with its entire package load. Formally, the interdicted logistics plan is hence a \enquote{truncated} plan
\begin{equation*}
\blueplan_c(\redplan) = 
    \begin{cases}
        \blueplan_c & \text{if~}\;\forall e\in \blueplan_c: \mathsf{E}(e)\not\in\redplan\\
        (e_{c,1},\dots, e_{c,j}) & \text{if~}\;\forall e\in (e_{c,1},\dots e_{c,j-1}): \mathsf{E}(e)\not\in\redplan\text{~and~}\;\mathsf{E}(e_{c,j})\in\redplan.
    \end{cases}
\end{equation*}
(ii) While the logistics plan (i.e., the connector paths) is fixed, the package flows are \textit{adaptive}, optimizing the Leontief value for the Blue's truncated plan. For a pair $(\blueplan,\redplan)$, the utility $u(\blueplan,\redplan)$ can hence be described as the following LP\footnote{Note the formulation is indeed an LP because the inner minimization in formulation~\eqref{leontief} is easily linearized using an auxiliary variable for each warehouse.}:
\begin{equation}
    \label{lp:recourse}
    \tag{U}
    u(\blueplan,\redplan) = \max_{s,l} v(\blueplan(\redplan),s,l) \quad
    \text{such that~} \text{R}(\blueplan(\redplan)) \text{~are satisfied},
\end{equation}
where $\text{R}(\cdot)$ refers to the set of flow constraints~\eqref{eq:R-constrs-first}-\eqref{eq:R-constrs-last}. Note that players cannot alter their strategies once they begin moving, making it, indeed, a one-shot game.  
We further assume the game is \textit{zero-sum}, i.e., Red minimizes $u$.

%% file: figs/physical-graph.tex
\begin{tikzpicture}[auto,node distance=2.5cm,
                    thick,main node/.style={circle,fill=\pcolor,draw,font=\sffamily\small\bfseries,inner sep=3, }, 
                    scale=0.4,minimum size=10pt]

  \node[main node] (A) at (0,0) {\supertiny{A}};
  
  \node[main node] (B) at (3,3) {\supertiny{B}};
  
  \node[main node] (C) at (3,0) {\supertiny{C}};
  
  


\foreach \source/\dest in {A/B, B/C, C/A}
     \draw[<->] (\source) -- (\dest);

\node at (0,2.6) {\scriptsize{$c_1$ only}};
\draw[->] (A) edge [loop above, min distance =22mm, out=70, in=110] (A);

\end{tikzpicture}

%% file: figs/unrolled-c1.tex
\begin{tikzpicture}[->,node distance=2.5cm,
                    thick,main node/.style={circle,fill=\gcolor,draw,font=\sffamily\small\bfseries,inner sep=3}, 
                    gray node/.style={circle,fill=\ncolor,draw,font=\sffamily\small\bfseries,inner sep=3},
                    scale=0.4,minimum size=10pt]

  
  \node[main node] (A1) at (3,2) {\supertiny{A}};
  \node[gray node] (B1) at (3,0) {\supertiny{B}};
  \node[gray node] (C1) at (3,-2) {\supertiny{C}};
  
  \node[gray node] (A2) at (6,2) {\supertiny{A}};
  \node[main node] (B2) at (6,0) {\supertiny{B}};
  \node[main node] (C2) at (6,-2) {\supertiny{C}};
  
  \node[main node] (A3) at (9,2) {\supertiny{A}};
  \node[main node] (B3) at (9,0) {\supertiny{B}};
  \node[main node] (C3) at (9,-2) {\supertiny{C}};
  


\draw  (A1.east) -- (B2.west);
\draw  (A1.east) -- (C2.west);
\draw  (A1.east) to[out=60, in=120] (A3.west);

\draw  (B2.east) -- (C3.west);
\draw  (B2.east) -- (A3.west);
\draw  (C2.east) -- (B3.west);
\draw  (C2.east) -- (A3.west);



\end{tikzpicture}

%% file: figs/unrolled-c2.tex
\begin{tikzpicture}[->,node distance=2.5cm,
                    thick,main node/.style={circle,fill=\gcolor,draw,font=\sffamily\small\bfseries,inner sep=3}, 
                    gray node/.style={circle,fill=\ncolor,draw,font=\sffamily\small\bfseries,inner sep=3},
                    scale=0.4,minimum size=10pt]

  
  \node[gray node] (A1) at (3,2) {\supertiny{A}};
  \node[main node] (B1) at (3,0) {\supertiny{B}};
  \node[gray node] (C1) at (3,-2) {\supertiny{C}};
  
  \node[main node] (A2) at (6,2) {\supertiny{A}};
  \node[gray node] (B2) at (6,0) {\supertiny{B}};
  \node[main node] (C2) at (6,-2) {\supertiny{C}};
  
  \node[main node] (A3) at (9,2) {\supertiny{A}};
  \node[main node] (B3) at (9,0) {\supertiny{B}};
  \node[main node] (C3) at (9,-2) {\supertiny{C}};
  


\draw  (B1.east) -- (C2.west);
\draw  (B1.east) -- (A2.west);

\draw  (A2.east) -- (B3.west);
\draw  (A2.east) -- (C3.west);
\draw  (C2.east) -- (B3.west);
\draw  (C2.east) -- (A3.west);



\end{tikzpicture}

%% file: sections/model.tex
\section{Computing Solutions of Contested Logistics}
Our goal is to find a Nash equilibrium (NE), possibly mixed, over Blue's logistics plans and Red's interdiction plans. Denote by $\Delta_b$ and $\Delta_r$ the probability simplices over $\blueplans$ and $\redplans$ respectively. Then, for some distribution over a player's plans $x_i \in \Delta_i$, $x_i(p_i)$ is the probability that $p_i$ is played by player $i\in\{r,b\}$. 
The NE problem reduces to solving the bilinear saddle point problem
\begin{align*}
    &\min_{x_b \in \Delta_b} \max_{x_r \in \Delta_r} \mathbb{E}_{\lambda \sim x_b, \iota \sim x_r} \left[ u (\lambda, \iota) \right] \label{eq:general-min-max}
    = &\min_{x_b \in \Delta_b} \max_{x_r \in \Delta_r} \sum_{\lambda \in \blueplans} \sum_{\iota\in\redplans} x_b(\lambda) \cdot x_r(\iota) \cdot u(\lambda, \iota). \nonumber
\end{align*}
Since $\Delta_b$ and $\Delta_r$ are both convex and compact sets, and the objective function is convex-concave, the minimax theorem \cite{v1928theorie} guarantees the existence of a unique value for the game. Nevertheless, determining the NE in CL games is computationally infeasible, as indicated already by the growth of the number of possible logistics plans that is double exponential in the problem's parameters.

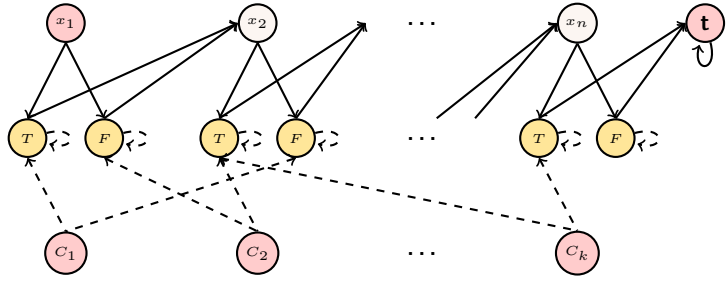
\begin{figure}[t]
\centering
\input{figs/eqm}
\caption{The physical graph described in Proposition~\ref{prop:eqm}, serving as a game for the 3-SAT problem we aim to reduce from. Each node in the top layer, denoted by $x_i$, corresponds to a variable $x_i$ in the SAT formula, which contains a total of $n$ variables. The nodes in the layer below signify a positive or negative assignment. Each assignment node is connected to the clauses it satisfies. For example, in the depicted graph $C_1=x_1\lor\neg x_2$, $C_2=\neg x_1\lor x_2$, and $C_n = x_2\lor x_n$. Edges available to the assignment connector starting from $x_1$ are solid, edges of the clause connectors starting from $C_j$ are dashed. Every path of the assignment connector of length $2n$ ending in the terminal node $t$ encodes a full assignment.}
\label{fig:sat}
\end{figure}

\begin{proposition}\label{prop:eqm}
It is \NP-hard in terms of $|\pgraph|$, $|\connectors|$, $|\packages|$, and $T$ to find a NE for a contested logistics game with Leontief utilities given in Formulation~\ref{lp:recourse}.
\end{proposition}
\begin{proof}
    We employ a reduction from the 3-SAT problem. Assume we are given a CNF having $n$ variables and $k$ clauses, where each clause has at most 3 literals. We aim to determine the satisfiability of the formula. We construct a CL scenario featuring a single type of package with unit weight and volume, and $k+1$ connectors: a single assignment connector constrained to weight and volume limits of $k$, and $k$ clause connectors with limits of $1$. The physical graph is depicted in Figure~\ref{fig:sat}. The top layer consists of a single node per each variable $x_i$, and a terminal node $t$. The layer below has two nodes per each $x_i$, signifying a positive (T) or a negative (F) assignment. The bottom layer contains one node per each clause. Each of the $x_i$ nodes is connected to its assignment nodes. Each clause node has an edge to a corresponding assignment node of every variable included in the clause. Moreover, the assignment nodes of the variable $x_i, i<n$ are connected to the variable node $x_{i+1}$. The assignment nodes of $x_n$ are connected to $t$. The first connector starts at node $x_1$. Other connectors start at their corresponding clause. Moving across each edge takes one time step. The outgoing edges from the assignment nodes are not available to the connectors starting at the clause nodes. The scenario's time horizon is $2n$. There are $2n+k+2$ warehouses: one at $x_1$, one at $t$, and one in every clause and assignment nodes. Only the clause warehouses supply a single unit of the package each. There is only one demand node, $t$, with demand $k$ and both unit payoff and maximum units set to 1. Red has budget $1$, with each edge having a cost of 2, except the loop in $t$ with cost $1$. Red's action space is hence trivial. 
    In the equilibrium, the assignment connector collects as many packages from the satisfied clauses as possible. We will show the value of the equilibrium is $1$ if and only if the formula is satisfiable.
    \begin{itemize}
        \item[$\rightarrow$] 
        Suppose there exists a satisfying assignment. Let the satisfying assignment define the path of the assignment connector. For each clause, there exists a literal that the assignment makes true. Let these literals define the paths of the clause connectors. From the definition of the satisfying assignment, the assignment connector's path crosses all the paths of the clause connectors, which enables it to pick up all $k$ packages and bring them to $t$. The Leontief utility at $t$ is hence $k/k=1$.
        \item[$\leftarrow$] If the value of the equilibrium is $1$, then the assignment connector must have picked up all $k$ packages, meeting with all $k$ clause connectors. Due to the construction of the physical graph, the assignment connector can visit only one of the assignment nodes for each variable, effectively encoding a variable assignment. Because the clause nodes are connected with only those assignment nodes that satisfy the clause, meeting with an assignment connector corresponds to satisfying the clause. The path of the assignment connector meeting all $k$ clause connector hence encodes a satisfying assignment.
    \end{itemize}
\end{proof}

\subsection{Best Response Complexity and Computation}

Recall that Blue's and Red's (pure) best responses to fixed strategies $x_r$ and $x_b$ are defined as $\blueplan^{\text{BR}} = \argmax_{\blueplan \in \blueplans} u(\blueplan, x_r)$ and $\redplan^{\text{BR}} = \argmin_{\redplan\in\redplans} u(x_b, \redplan)$, where $u(x_b, \redplan)=\mathbb{E}_{\blueplan \sim x_b} u(\blueplan, \redplan)$ and $u(\blueplan, x_r)=\mathbb{E}_{\redplan \sim x_r} u(\blueplan, \redplan)$. We note that while best responses are closely related to NE computation, they are generally distinct problems. There are classes of games where computing best responses is difficult but finding a NE is easy, and vice versa \cite{xu2016mysteries}. Unfortunately, computing best responses in CL games is also intractable. 
Indeed, intractability of Blue's best response follows directly from the proof construction of Proposition~\ref{prop:eqm}.

\begin{corollary}\label{prop:sat}
    Let $\tilde{\redplans} \subseteq \redplans$ be of size $k$ (possibly smaller than $|\redplans|$) and $\widetilde{x}_r$ be a distribution with support $\tilde{\redplans}$.
    Finding Blue's best response against $\widetilde{x}_r$ in a CL problem with Leontief utilities is \NP-hard in terms of $|\pgraph|$, $|\connectors|$, $|\packages|$, $T$, and $k$. 
\end{corollary}  

Computing Blue's best response can be done via a polynomially sized MILP. Assume Red plays the interdiction plans $\redplan^1,\dots,\redplan^k$ with probabilities $x_r^1,\dots,x_r^k$. For each $\redplan^i$, let us denote the set of Blue's edges in their layered graph that are interdicted by $\redplan^i$ as $\redplan^i_c=\{e\in\ledges_c:\;\mathsf{E}(e)\in\redplan^i\}$. The the best response is formulated as the following max-max \textsc{BlueBR} formulation:

\vspace{-.3cm}
\begin{align*}
    \max_{f\in\ref{blue-flows}}\max_{l,s,g}\;&\sum_{i\in [k]}\sum_{w\in\warehouses} P(w)x_r^i g^i_w\\
        S(w,p) &= s^i_{w_0,p}&&\forall i\in[k],\;\forall w\in\warehouses,\;\forall p\in\packages\\
    \smashoperator{\sum_{e\in\ledges^+_c(v_t)\backslash\redplan^i_c}}l^i_{c,p}(e) &= \smashoperator{\sum_{e\in\ledges^-_c(v_t)}}l^i_{c,p}(e) && \forall i\in[k],\;\forall c\in\connectors,\;\forall p\in\packages,\\
    & &&\forall t\in\timesteps\backslash\{0,T\},\;\forall v\not\in\warehouses\\
    l^i_{c,p}(e) &\leq M f_{c}(e) && \forall i\in[k],\;\forall c\in\connectors,\;\forall p\in\packages,\;\forall e\in\ledges_c\\
    W_{max}(c)&\geq \sum_{p\in\packages} W(p)l^i_{c,p}(e) &&\forall i\in[k],\;\forall c\in \connectors,\;\forall e\in \ledges_c \\
     V_{max}(c)&\geq \sum_{p\in\packages} V(p)l^i_{c,p}(e) &&\forall i\in[k],\;\forall c\in \connectors,\;\forall e\in \ledges_c\\
     s^i_{w_t,p} + \smashoperator{\sum_{c\in\connectors,e\in\ledges^+_c(w_t)\backslash\redplan^i_c}}l^i_{c,p}(e) &= s^i_{w_{t+1},p} + \smashoperator{\sum_{c\in\connectors,e\in\ledges^-_c(w_t)}}l^i_{c,p}(e)  && \forall i\in[k],\;\forall t\in\timesteps,\;\forall p\in\packages,\;\forall w\in\warehouses\\  \smashoperator{\sum_{\substack{p\in\packages,t\in\timesteps,w\in\warehouses \\ e'\in\ledges^-_c(w_t): e\subset e'}}}l^i_{c,p}(e') &\leq M\cdot (1-f_{c}(e))&&\forall i\in[k],\;\forall c\in \connectors,\;\forall e\in\redplan_c^i\\
    0\leq g^i_w &\leq s^i_{w_{T+1},p} / D(w,p) && \forall i\in[k],\;\forall w\in\warehouses,\; \forall p\in\packages: D(w,p) > 0 \\
         l^i_{c,p}(e)&\geq 0 && \forall i\in[k],\;\forall c\in \connectors,\;\forall p\in \packages,\;\forall e\in\ledges_c\\
        s^i_{w_t,p}&\geq 0 && \forall i\in[k],\;\forall t\in \timesteps\cup\{T+1\},\\
        & && \forall p\in \packages,\;\forall w\in \warehouses\\
    g^i_w &\leq U(w) && \forall i\in[k],\;\forall w\in\warehouses.
\end{align*}

Note that in this MILP, for each $i\in [k]$, we have a different load flow $l^i_{c,p}(e)$. To simulate the truncated logistics plans, any potentially positive load on an interdicted edge is omitted as an \textit{incoming} load from the conservation constraints in the following node. Moreover, we need to make sure that if a connector of a particular logistics plan gets destroyed, its load stays zero for all future time steps, especially if its path is scheduled to cross a warehouse. This is achieved by the second, \textit{load-cancelling} big-M constraint. Here, by $e\subset e'$ we denote for edges $e\neq e'\in\ledges_c$ that $e'$ is reachable in the layered graph from $e$.

\begin{figure}[t]
\centering
\input{figs/br_red}
\caption{The physical graph described in Proposition~\ref{prop:br_red}, serving as a game for the set cover problem we aim to reduce from. The second and third layers are identical, containing a node for each set $S_i$. Every $u_i$  corresponds to a path going through the edges $S_j,S'_j$ of all the sets $S_j$ the $u_i$ is contained in. Red can interdict only the forward edges between the second and third layer (depicted in red color), encoding a selection of sets in the cover.
}
\label{fig:br_red}
\end{figure}
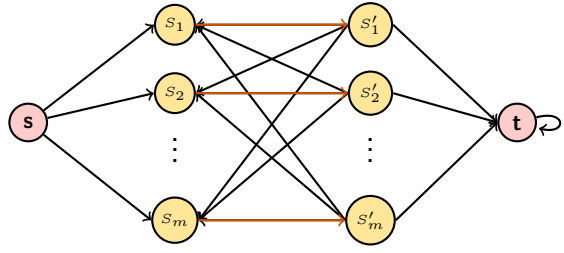

\begin{proposition}\label{prop:br_red}
Let $\tilde{\blueplans} \subseteq \blueplans$ be of size $k$ (possibly smaller than $|\blueplans|$) and $\widetilde{x}_b$ be a distribution with support $\tilde{\blueplans}$. Finding Red's best response against $\widetilde{x}_b$ in a contested logistics problem with Leontief utilities is \NP-hard in terms of $|\pgraph|$, $b$, and $k$.
\end{proposition}
\begin{proof}
    We reduce the set cover problem with the universe $\mathcal{U} = \{u_1, u_2, \dots, u_n\}$, a collection of sets $S = \{S_1, S_2, \dots, S_m\}$, and an integer $b$ to a CL scenario. This scenario features a single package type with unit weight and volume, and a single connector with weight and volume limits of 1. The connector moves across a graph with four layers: the first and last contain only nodes $s$ and $t$, respectively. The second and third layers are identical, each with one node per set $S_j$. Node $s$ connects to each $S_i$ in the second layer. The second and third layers are connected by edges between corresponding nodes $S_j$ and $S'_j$. Each $S'_j$ in the third layer connects to every $S_i$ in the second layer and to $t$. Node $t$ has a loop to itself. The connector starts at $s$, each edge takes one time step to cross, and the time horizon is $2n + 1$. There are warehouses at nodes $s$ and $t$, with a supply of 1 at $s$ and a demand of 1 at $t$, each with a payoff and maximum of 1. Red has a budget of $b$ and can interdict only the edges between the second and third layers, each costing 1. Blue's mixed strategy is constructed as follows: for each $u_i$, create $T_i = \{S_j \mid u_i \in S_j\}$, take an arbitrary enumeration $(t_1,\dots,t_k)$ of $T_i$, and define a path 
    $P_i = (s, t_1, t'_1,\dots,t_k,t'_k, t,\dots,t)$ with $2n-2|T_i|+1$ loops in $t$ at the end, 
    each played with equal probability. Every $u_i$ hence corresponds to a path going through all the sets $u_i$ is contained in, terminated by loops in $t$. Note that the order in which the sets in $T_i$ are traversed in $P_i$ does not matter. The BR value is 0 if and only if there is a set cover of size at most $b$.
    \begin{itemize}
        \item[$\rightarrow$] If the value is 0, then there exists a best response strategy that interdicts all $n$ paths. From the construction, the size of this interdiction plan is at most $b$, and the selected edges encode a selection of at most $b$ sets from $S$. Because every path corresponds to one $u_i$, the selected sets form a set cover.
        \item[$\leftarrow$] is analogous. 
    \end{itemize} 
\end{proof}

In practice, we optimize Red's utility using the feasibility formulation~\eqref{red-interdiction}. Assume Blue plays the logistic plans $\blueplan^1,\dots, \blueplan^k$ with probabilities $x_b^1,\dots,x_b^k$. The optimal Red's interdicting plan $\redplan$, encoded via binary indicators $y$, can be formulated as the following simple min-max MILP, with the flow constraints~(R) for each logistics plan in the support:

\begin{equation*}
\label{red-br}
\begin{aligned}
    \min_{y\in\ref{red-interdiction}}\max_{s,l,g}&\sum_{i\in [k]}\sum_{w\in\warehouses} P(w)x_b^i g^i_w - &&\sum_{i\in [k]}\sum_{c\in\connectors}\sum_{e^i_{c,j}\in \blueplan^i_c}\sum_{\substack{e^i_{c,k}\in \blueplan^i_c\\ k\geq j}}\sum_{p\in\packages}Zy(\mathsf{E}(e^i_{c,j}))l^i_{c,p}(e^i_{c,k})\\
    & \text{R}(\blueplan^i)&&\forall i\in[k]\\
    0\leq g^i_w &\leq s^i_{w_{T+1},p} / D(w,p) &&\forall i\in[k],\;\forall w\in\warehouses,\; \forall p\in\packages: D(w,p) > 0 \\
    g^i_w &\leq U(w) &&\forall i\in[k],\;\forall w\in\warehouses.
\end{aligned}
\end{equation*}
Note the penalty term in the objective that plays a similar role to the load cancelling constraint in Blue's BR. Using penalty terms is less numerically stable than the big-M constraints. However, using the same approach as in the \textsc{BlueBR} would result in bilinear terms in the constraints, that are more cumbersome to linearize, and involve unbounded big-M constants. The constant $Z$ is chosen to make any potential increase in the Leontief utility that Blue might gain by sending a positive load over an interdicted edge undesirable due to the incurred penalty, e.g., $Z=\max_{w\in\warehouses} P(w)$.
Since the inner problem is an LP, we can dualize it, which removes the bilinear terms in the objective and gives us the final \textsc{RedBR} integer formulation. Due to space constraints, we have deferred the specific details of this dualization to the extended version of the paper.

\subsection{Approximating NE Using Strategy Generation}

Despite the exponential size of Blue's strategy space, practical CL scenarios often exhibit equilibria with relatively small supports. This observation leads us to employ the \textit{double oracle} (DO) framework.

The DO algorithm (Algorithm~\ref{alg:do}) is an iterative, specialized form of concurrent column and row generation. It is frequently used to address large saddle-point problems that have efficient (in practical terms) best-response oracles. The DO algorithm incrementally constructs a subgame -- a subset of pure strategies for each player -- with the intention of excluding strategies that do not contribute to the equilibrium. At the conclusion of the algorithm, the subgame (ideally a small portion of the entire game) contains a NE that mirrors the NE of the original game. In our context, pure strategies consist of logistic and interdiction plans, and subgames are defined by subsets $\widetilde{\blueplans} \subseteq \blueplans$ and $\widetilde{\redplans} \subseteq \redplans$. 

The process begins with a small subgame for each player, $\widetilde{\blueplans}$ and $\widetilde{\redplans}$. In each iteration, it calculates the equilibrium $(\widetilde{x}^*_b, \widetilde{x}^*_r)$ within the current subgame, allowing players to choose distributions of plans only from $\widetilde{\blueplans}$ or $\widetilde{\redplans}$. For each player $i$, we determine the best responses $\blueplan^{\text{BR}}$ and $\redplan^{\text{BR}}$ against their opponent's subgame equilibrium strategy $\widetilde{x}_{-i}^*$, using best-response oracles. These best responses introduce new plans into the subgame, and the process repeats.

The DO algorithm terminates when the best-response oracles produce responses that do not enhance any of the player's utility over the subgame value. This indicates that the current subgame equilibrium is also an equilibrium in the full game, and adding more strategies will not yield less exploitable strategies for either player. In practice, instead of converging to an exact equilibrium, we calculate the equilibrium gap $\nabla = u(\widetilde{x}_b^*, \redplan^\text{BR}) - u(\blueplan^\text{BR}, \widetilde{x}_r^*)$ and terminate when $\nabla \leq \epsilon$ for a predetermined threshold $\epsilon > 0$, returning a $2\epsilon$-approximate-NE.

In practice, the time needed to determine Blue's best response significantly affects the overall runtime. To accelerate the computation, we set a predetermined time limit for solving the MILP, rather than solving it to full completion. This approach yields an approximate best response that is generally close to the optimal solution. Periodically, and before the final termination, we solve Blue's BR MILP to completion to ensure the equilibrium gap is computed accurately.

\begin{algorithm}[t]
\caption{Double Oracle for Contested Logistics Games}\label{alg:do}
\begin{algorithmic}[1]
\STATE $\widetilde{\blueplans}, \widetilde{\redplans} \gets \textsc{InitialSubgame}(\blueplans, \redplans)$
\REPEAT
\STATE $\widetilde{x}_b^*, \widetilde{x}_r^* \gets \textsc{NashEquilibrium}(\widetilde{\blueplans}, \widetilde{\redplans})$
\STATE $\blueplan^\text{BR}, \redplan^\text{BR} \gets \textsc{BlueBR}(\widetilde{x}_r^*), \textsc{RedBR}(\widetilde{x}_b^*)$ 
\STATE $\widetilde{\blueplans}, \widetilde{\redplans} \gets \widetilde{\blueplans}\cup \{\blueplan^\text{BR}\}, \widetilde{\redplans}\cup \{\redplan^\text{BR}\}$
\UNTIL{$\textsc{EquilibriumGap}(\widetilde{x}_b^*, \widetilde{x}_r^*, \blueplan^\text{BR}, \redplan^\text{BR}) \leq \epsilon$}
\end{algorithmic}
\end{algorithm}

%% file: figs/eqm.tex
\begin{tikzpicture}[->,node distance=2.5cm,
                    thick,main node/.style={circle,fill=\gcolor,draw,font=\sffamily\small\bfseries,inner sep=2,minimum size=.5cm}, 
                    gray node/.style={circle,fill=\ncolor,draw,font=\sffamily\small\bfseries,inner sep=2,minimum size=.5cm},
                    st node/.style={circle,fill=red!20,draw,font=\sffamily\small\bfseries,inner sep=2,minimum size=.5cm},
                    scale=0.17,minimum size=10pt]

\pgfdeclarelayer{background}
\pgfdeclarelayer{foreground}
\pgfsetlayers{background,main,foreground}

\def\lidepth{9}
\def\liidepth{0}
\def\liiidepth{-9}
  
  
  \node[st node] (X1) at (5, \lidepth) {\supertiny{$x_1$}};
  \node[gray node] (X2) at (20, \lidepth) {\supertiny{$x_2$}};
  \node[gray node,draw=none,fill=none] (X) at (30, \lidepth) {};
  \node[main node,fill=none,draw=none] (XX) at (33,\lidepth) {\dots};
  \node[gray node] (Xn) at (45, \lidepth) {\supertiny{$x_n$}};
  \node[st node] (t) at (55, \lidepth) {t};
  
  \node[main node] (X1t) at (2,\liidepth) {\supertiny{$T$}};
  \node[main node] (X1f) at (8,\liidepth) {\supertiny{$F$}};
  \node[main node] (X2t) at (17,\liidepth) {\supertiny{$T$}};
  \node[main node] (X2f) at (23,\liidepth) {\supertiny{$F$}};
  \node[main node,draw=none,fill=none] (Xt) at (34,\liidepth) {};
  \node[main node,draw=none,fill=none] (Xf) at (37,\liidepth) {};
  \node[main node,fill=none,draw=none] (Xd) at (33,\liidepth) {\dots};
  \node[main node] (Xnt) at (42,\liidepth) {\supertiny{$T$}};
  \node[main node] (Xnf) at (48,\liidepth) {\supertiny{$F$}};
  
  \node[st node] (C1) at (5, \liiidepth) {\supertiny{$C_1$}};
  \node[st node] (C2) at (20, \liiidepth) {\supertiny{$C_2$}};
  \node[main node,fill=none,draw=none] (Cd) at (33,\liiidepth) {\dots};
  \node[st node] (Cn) at (45, \liiidepth) {\supertiny{$C_{k}$}};

  

\draw (X1.south) -- (X1t.north);
\draw (X1.south) -- (X1f.north);
\draw (X2.south) -- (X2t.north);
\draw (X2.south) -- (X2f.north);
\draw (Xn.south) -- (Xnt.north);
\draw (Xn.south) -- (Xnf.north);



\draw (X1t.north) -- (X2.west);
\draw (X1f.north) -- (X2.west);
\draw (Xnt.north) -- (t.west);
\draw (Xnf.north) -- (t.west);

\draw (X2t.north) -- (X.west);
\draw (X2f.north) -- (X.west);
\draw (Xt.north) -- (Xn.west);
\draw (Xf.north) -- (Xn.west);

\draw[<-, dashed] (X1t.south) -- (C1.north);
\draw[<-, dashed] (X2f.south) -- (C1.north);
\draw[<-, dashed] (X1f.south) -- (C2.north);
\draw[<-, dashed] (X2t.south) -- (C2.north);
\draw[<-, dashed] (Xnt.south) -- (Cn.north);
\draw[<-, dashed] (X2t.south) -- (Cn.north);






\path (t) edge[loop below] ();

\path[dashed] (X1t) edge[loop right] ();
\path[dashed] (X1f) edge[loop right] ();
\path[dashed] (X2t) edge[loop right] ();
\path[dashed] (X2f) edge[loop right] ();
\path[dashed] (Xnt) edge[loop right] ();
\path[dashed] (Xnf) edge[loop right] ();



\end{tikzpicture}

%% file: figs/br_red.tex
\begin{tikzpicture}[->,node distance=2.5cm,
                    thick,main node/.style={circle,fill=\gcolor,draw,font=\sffamily\small\bfseries,inner sep=2,minimum size=.5cm}, 
                    gray node/.style={circle,fill=\ncolor,draw,font=\sffamily\small\bfseries,inner sep=2,minimum size=.5cm},
                    st node/.style={circle,fill=red!20,draw,font=\sffamily\small\bfseries,inner sep=2,minimum size=.5cm},
                    scale=0.13,minimum size=10pt]

\pgfdeclarelayer{background}
\pgfdeclarelayer{foreground}
\pgfsetlayers{background,main,foreground}

\def\liidepth{15}
\def\liiidepth{35}
\def\livdepth{50}
  
  \node[st node] (s) at (0,0) {{s}};
  
  
  \node[main node] (A2) at (\liidepth,10) {\supertiny{$S_1$}};
  \node[main node] (B2) at (\liidepth,3) {\supertiny{$S_2$}};
  \node[main node,fill=none,draw=none] (Xd) at (\liidepth,-2) {$\vdots$};
  \node[main node] (E2) at (\liidepth,-10) {\supertiny{$S_m$}};
  
  \node[main node] (A3) at (\liiidepth,10) {\supertiny{$S'_1$}};
  \node[main node] (B3) at (\liiidepth,3) {\supertiny{$S'_2$}};
    \node[main node,fill=none,draw=none] (Yd) at (\liiidepth,-2) {$\vdots$};
  \node[main node] (E3) at (\liiidepth,-10) {\supertiny{$S'_m$}};

  \node[st node] (t) at (\livdepth,0) {{t}};
  

\draw (s) -- (A2.west);
\draw (s) -- (B2.west);
\draw (s) -- (E2.west);



\draw  (A3.west) -- (B2.east);
\draw  (A3.west) -- (E2.east);
\draw  (B3.west) -- (A2.east);
\draw  (B3.west) -- (E2.east);
\draw  (E3.west) -- (A2.east);
\draw  (E3.west) -- (B2.east);

\draw[<-]  (A2.east)  -- (A3.west);
\draw[<-]  (B2.east)  -- (B3.west);
\draw[<-]  (E2.east)  -- (E3.west);

\def\redcolor{burntorange}
\draw[->]  (A2.east) [color=\redcolor] -- (A3.west);
\draw[->]  (B2.east) [color=\redcolor] -- (B3.west);
\draw[->]  (E2.east) [color=\redcolor] -- (E3.west);

\draw  (A3.east) -- (t.west);
\draw  (B3.east) -- (t.west);
\draw  (E3.east) -- (t.west);

\path (t) edge[loop right] ();



\end{tikzpicture}

%% file: sections/experiments.tex
\section{Empirical Evaluation}
Now we move to the experiments on contested logistics scenarios. Our goals are (i) to explore qualitatively how optimal strategies behave in real-world scenarios, and (ii) to evaluate the scalability of our proposed double oracle algorithm using synthetically generated maps. 

All experiments were conducted on an Intel Xeon Gold 6226 (2.9Ghz), restricted to 8 threads and equipped with 32GB of RAM. The (MI)LPs were solved with the Gurobi Optimizer version 10.0.3, build v10.0.3rc0 \cite{gurobi}, on a Linux 64-bit platform. The double oracle algorithm was implemented in Python 3.7.9, using a tolerance setting of $\epsilon=10^{-2}$, and 5s time limit for the MILP solver. 

\begin{figure}[t]
    \centering
\begin{subfigure}[t]{0.32\linewidth}
    \centering
    \input{figs/grid_5x5_uni}
\end{subfigure}
\hfill
\begin{subfigure}[t]{0.32\linewidth}
    \centering
    \input{figs/grid_6x6_uni}
\end{subfigure}
\hfill
\begin{subfigure}[t]{0.32\linewidth}
    \centering
    \input{figs/grid_7x7_uni}
\end{subfigure}
    \caption{Computation times of the double oracle algorithm for grid world contested logistics scenarios with \textit{uniform} edge interdiction costs.}
    \label{fig:runtimes:uniform}
\end{figure}
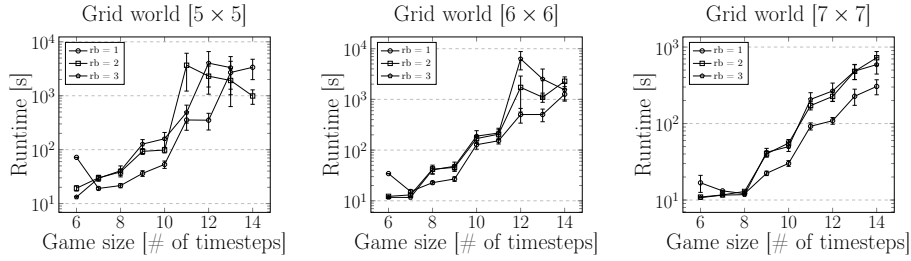

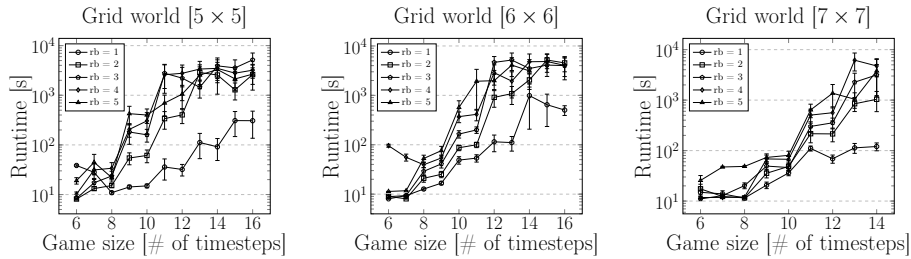
\begin{figure}[t]
    \centering
\begin{subfigure}[t]{0.32\linewidth}
    \centering
    \input{figs/grid_5x5_runtimes}
\end{subfigure}
\hfill
\begin{subfigure}[t]{0.32\linewidth}
    \centering
    \input{figs/grid_6x6_runtimes}
\end{subfigure}
\hfill
\begin{subfigure}[t]{0.32\linewidth}
    \centering
    \input{figs/grid_7x7_runtimes}
\end{subfigure}
    \caption{Computation times of the double oracle algorithm for grid world contested logistics scenarios with \textit{randomly assigned} edge interdiction costs.}
    \label{fig:runtimes}
\end{figure}

\subsection{Quantitative Evaluation on a Synthetic Grid World}

First, we evaluate the performance of the DO on simple grid world scenarios. In these scenarios, the physical graph consists of an \(N \times N\) grid. Blue designates all four corners as warehouses, with an additional warehouse located at the center of the grid. For even values of \(N\), the central warehouse is one of the four central nodes. Additionally, Blue has two trucks, initially positioned at opposing corners. These trucks have sufficient weight and volume limits to transport any available packages, and they can move along any single edge per time step. There are two types of packages, A and B, each with unit weights and volumes. The warehouse at the initial location \((0,0)\) of the first truck holds 4 units of A and 1 unit of B. The warehouse at the location \((N-1, N-1)\) of the other truck holds 1 unit of A and 3 units of B. The central warehouse supplies a single unit of each package. Only two warehouses have positive demands, located in the remaining two corners without the trucks. Both warehouses require 3 A units and 2 B units.

To generate random grid world scenarios, each edge is removed with a probability of 0.1. Each warehouse with a demand is assigned a uniformly random real payoff from the interval \([1,2]\). Unless the edges have a uniform cost, the cost is selected uniformly from the integer interval \([1,5]\). For statistical robustness, 20 instances of each game were constructed and solved, with average results reported alongside standard errors. Examples of the initial setup can be seen in the two physical graphs in Figure~\ref{fig:heatmaps}. In these maps, the starting locations of the two trucks are shown as blue nodes, while the demand nodes are orange. The central warehouse is purple, and each edge is annotated with its interdiction cost.

Figure~\ref{fig:runtimes:uniform} shows the average runtimes for the DO to solve grid scenarios of sizes \(5 \times 5\), \(6 \times 6\), and \(7 \times 7\), each with uniform edge interdiction costs of 1. Solving the game for Red's budget of 1 (denoted as \(rb=1\)) is clearly the easiest, with higher budgets presenting similar levels of difficulty. Notably, the game becomes trivial for a budget of 4, when the trucks can be completely cut off from reaching the demand nodes. In Figure~\ref{fig:runtimes}, where interdiction costs are random, the difference between budget 1 and higher budgets becomes even more pronounced.

To illustrate how the game value changes with Blue's horizon and Red's capabilities, we selected two typical scenarios and depicted the generated physical graphs, along with the resulting values as heatmaps, in Figure~\ref{fig:heatmaps}. In the first map, Red is able to completely cut off the connectors with a budget of 4, resulting in a game value of zero for Blue. In contrast, Red can more efficiently interdict Blue only if they possess a higher budget, and Blue's shorter horizon does not provide enough additional maneuverability to make a difference.

\begin{figure}[t]
    \centering
\begin{subfigure}[t]{0.24\linewidth}
    \centering
    \hspace*{-.3cm}\includegraphics[width=1.2\linewidth]{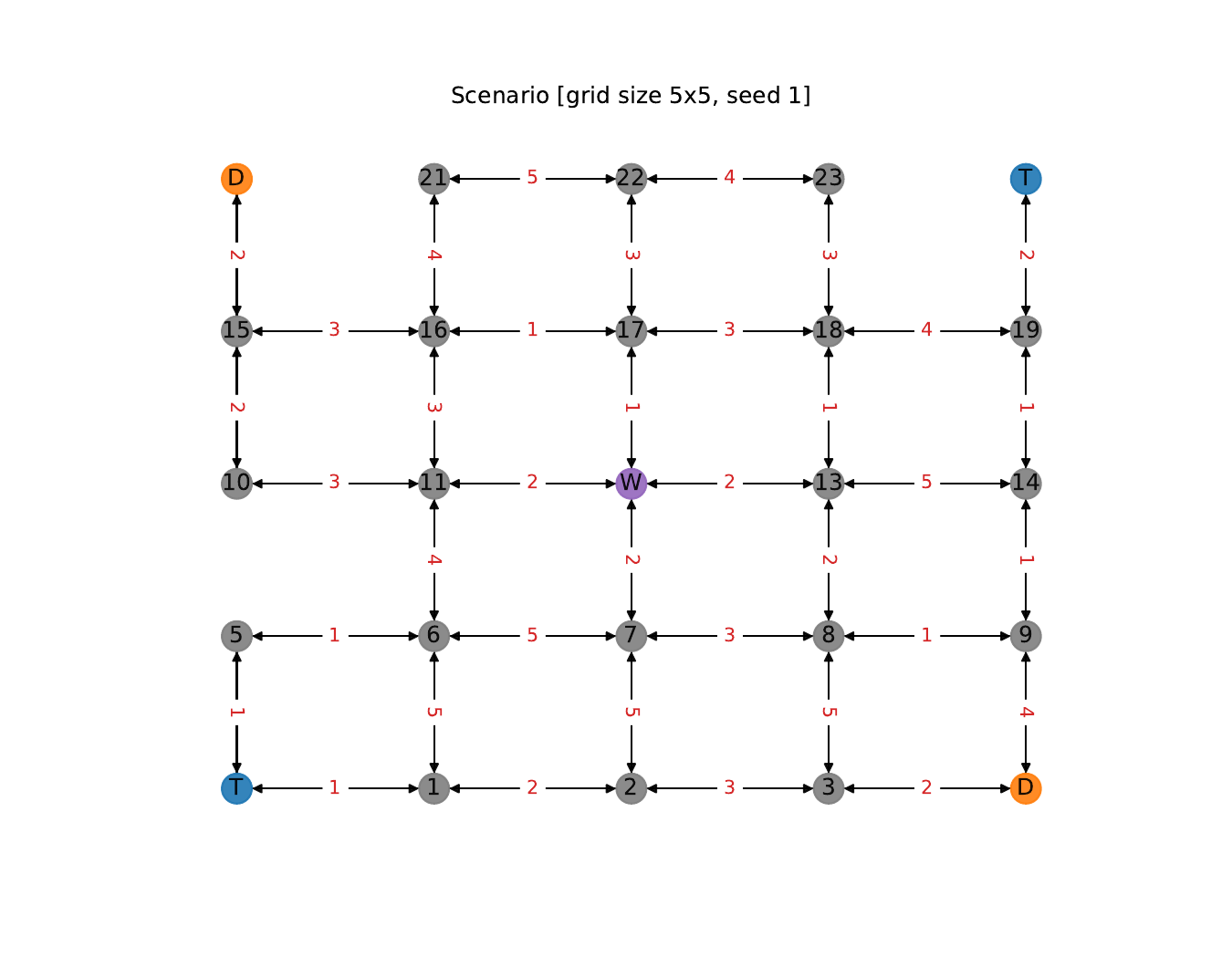}
\end{subfigure}
\hfill
\begin{subfigure}[t]{0.24\linewidth}
    \centering
    \includegraphics[width=.99\linewidth]{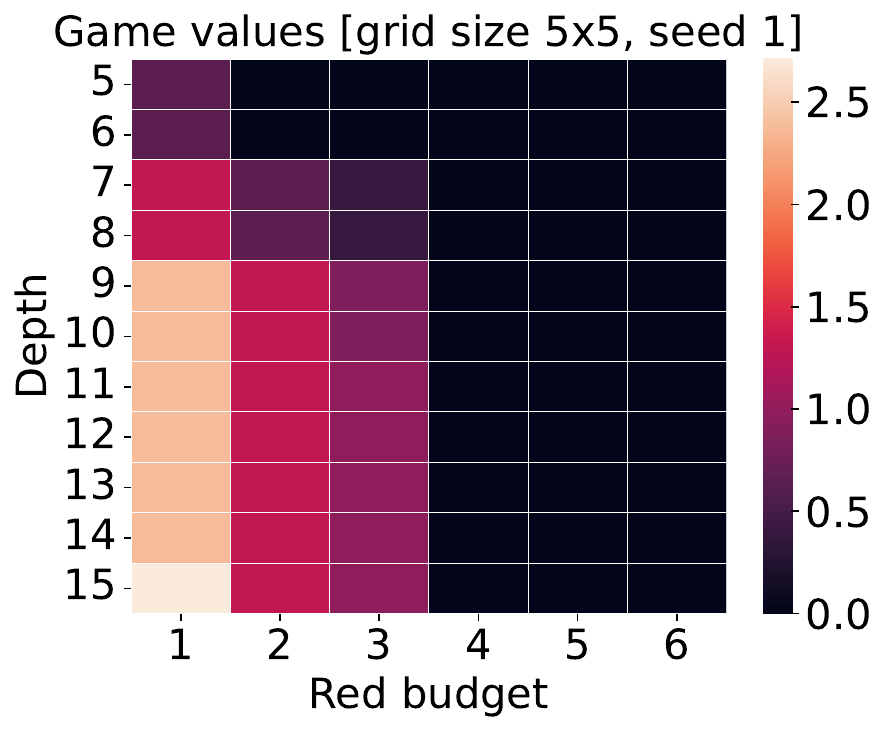}
\end{subfigure}
\hfill
\begin{subfigure}[t]{0.24\linewidth}
    \centering
    \hspace*{-.23cm}\includegraphics[width=1.2\linewidth]{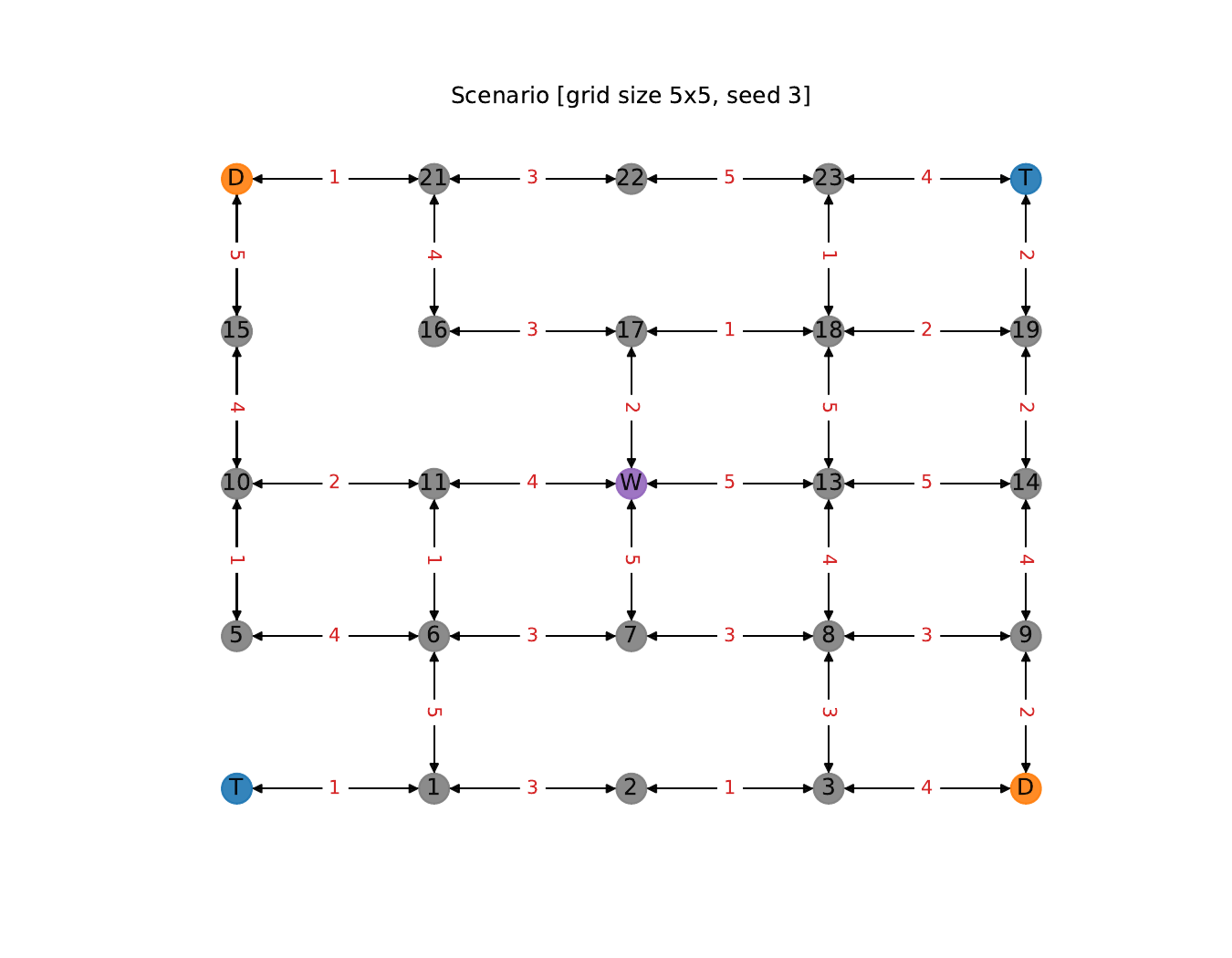}
\end{subfigure}
\hfill
\begin{subfigure}[t]{0.24\linewidth}
    \centering
    \includegraphics[width=.99\linewidth]{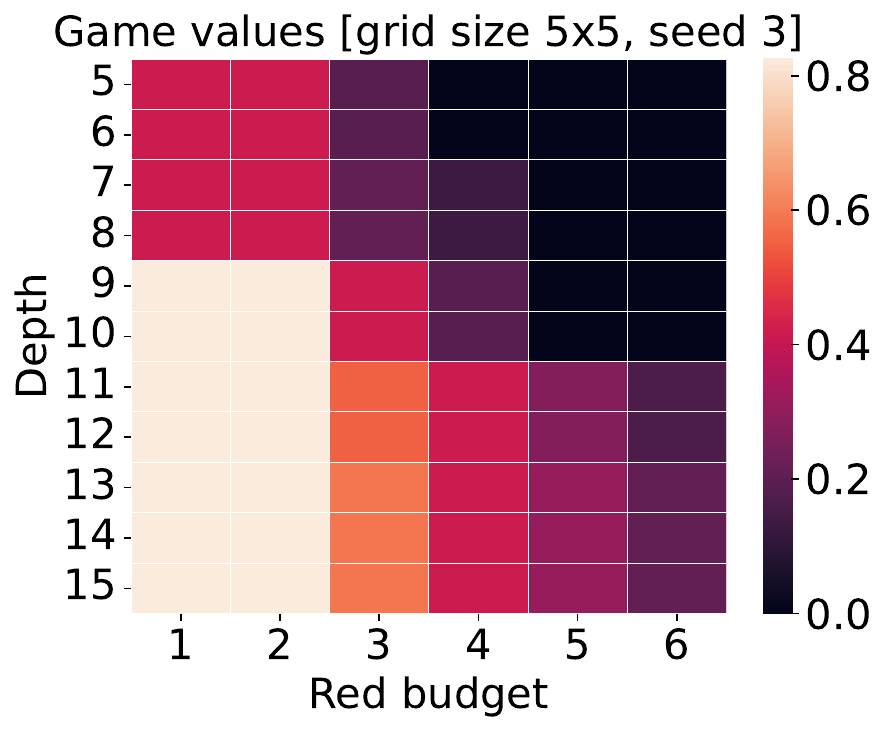}
\end{subfigure}
    \caption{The physical graphs for two random grid world contested logistics scenarios, and game values as functions of red's budget and blue's horizon. Each grid world has size $5\times 5$. Their corresponding game value heat maps (on the right) have Red's budget on the horizontal axis and time horizon on the vertical axis. Lighter color signifies higher game value.
    }
    \label{fig:heatmaps}
\end{figure}

\subsection{Qualitative Evaluation on Real-World Maps}

We also conducted experiments on 2 different maps around the world simulating CL scenarios. These are from (i) the United Kingdom (UK) based on railroads during World War 2, and (ii) Mariupol, a city heavily involved in the ongoing Russo-Ukrainian conflict. The goal of these experiments is not to evaluate the efficiency of our DO method, but rather to showcase (a) how real-world red and blue strategies will look like in practice, (b) the importance of strategic behavior (i.e., randomization) in both Red and Blue, and (c) the relatively low cost that blue pays to be robust to adversarial behavior, and conversely the extremely poor performance when ignoring existence of Red or when using heuristic solutions.

\subsubsection{Contested Logistics in the United Kingdom}
\begin{figure}[ht]
    \includegraphics[width=.99\linewidth]{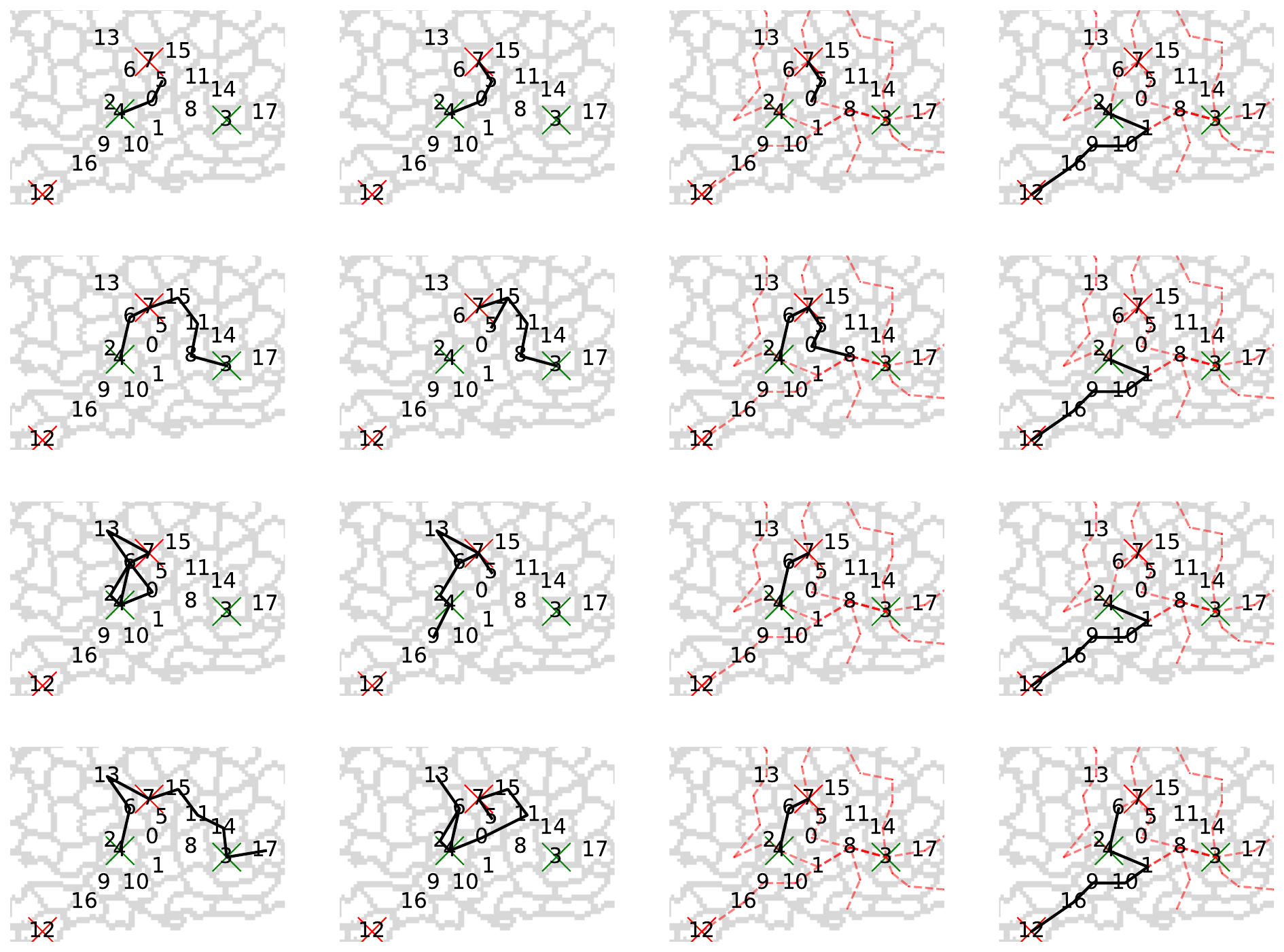}
    \caption{NE for Blue in the UK scenario. Each row corresponds to a logistics plan played with positive probabilities $4/9, 1/9, 5/21, 13/63$ respectively. 
    Light gray lines demarcate boundaries between provinces, numbers denote province labels: for simplicitly we have only included those encountered by some connector logistics plan. 
    Each column shows the movement of a single connector, in the order of Truck 1, Truck 2, Train 1, Train 2, which start at provinces 4, 5, 7, and 1 respectively.
    Connector paths are denoted by the solid black line line. Red crosses at 7 and 12 denote supply nodes, green crosses at 3 and 4 denote demand. Dotted lines denote railroads. }
    \label{fig:uk_blue_strats}
\end{figure}
\begin{figure}[ht]
    \includegraphics[width=.99\linewidth]{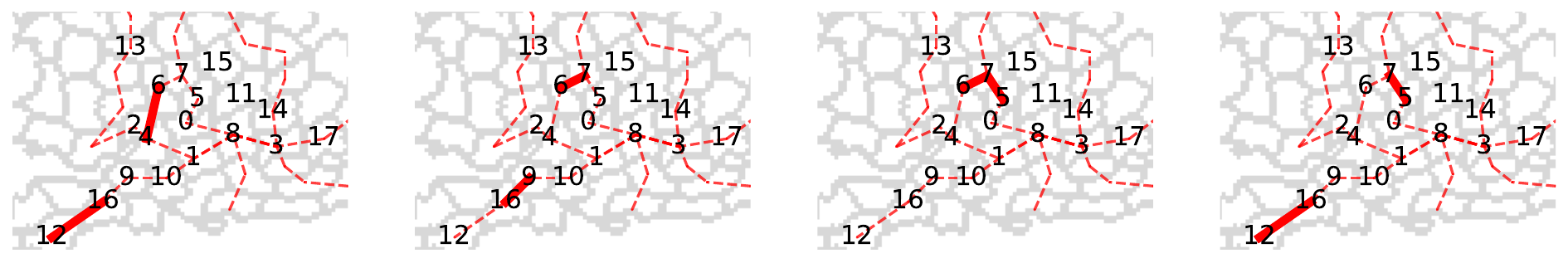}
    \caption{NE for Red in the UK scenario. Each subfigure corresponds to a single interdiction plan played with probabilities $7/18, 1/9, 7/18, 1/9$ respectively. Dotted lines denote railroads and thick red lines edges that are interdicted.}
    \label{fig:uk_red_strats}
\end{figure}
This scenario is based on the southern region of the United Kingdom. The region is broken into provinces, using data from the World War 2 based video game \textit{Heart of Iron IV} \cite{hoi4}. There are three essentially identical packages: boots, rifles and helmets; each soldier requires 1 unit of each to be equipped (note that we allow for ``fractional soldiers''). 
There are 4 connectors, comprising 2 trains and 2 trucks. Trucks can move between any two adjacent provinces and have a capacity of 5. Trains can only move between provinces connected by railroads, but enjoy a capacity of 20 (for this scenario, weight and size are identical quantities). There are two demand and supply nodes. Each supply node contains 20 of each package. There is unlimited demand for soldiers at each demand node. Red is able to interdict edges between any two adjacent provinces (recall that these are \textit{directed} edges) and has a budget of $2$. We set a time horizon of $10$ for Blue. The equilibrium strategies for Blue and Red are shown in Figures~\ref{fig:uk_blue_strats} and \ref{fig:uk_red_strats}, which has a Nash value of 9.259.
We discuss the most interesting aspects of the NE. 
\begin{itemize}
\item In all 4 logistics plans, Train 2 behaves essentially deterministically --- first, collect supply from 12, then deliver it to the demand at 4. Therefore, Red may interdict anywhere along this path (e.g., edge $16\rightarrow12$) and interdict Train 2 \textit{with certainty}. Surprisingly, we find that in interdiction plan 3, Red declines this ``freebie'', choosing instead to interdict edges $7 \rightarrow 6$ and $7 \rightarrow 5$.
\item Second, we observe that Train 1 coordinates with Trucks in a way such as to maximize ``throughput''. In logistics plan 3 and 4, we observe that Train 1 (which begins at a supply node 7) delivers directly to the demand node at 4 by moving along the left path comprising $7\rightarrow 6 \rightarrow 2$ and backward. Note that compared to Train 2, the distances between supply and demand is much shorter, so the throughput here is comparatively much higher. Blue also diversifies via logistics plan 1, where Train 1 moves to the right along $7 \rightarrow 5 \rightarrow 0$ instead. It then drops off supplies along the way at province 0, while allowing Trucks 1 and 2 to complete the ``last mile delivery'' to province 4. 
This hedges against Red always interdicting the left path $7 \rightarrow 6 \rightarrow 2$. This explain interdiction plan 3: by interdicting $7 \rightarrow 6$ and $6 \rightarrow 4$, Red completely shuts down the joint operation between Train 1 and the Trucks. 
\item We again observe hedging behavior in Logistics plan 2. Train 1 first takes the right path $7 \rightarrow  5 \rightarrow  0 \rightarrow 8$, drops off its packages (leaving last mile deliveries to Trucks), returns to 7 to pick up fresh supplies, and finally takes the left path  $7\rightarrow  6 \rightarrow  4$ to satisfy the demand at 4. Interestingly, this exposes Train 1 to interdictions both on the left and right. However, if interdiction is on the left (e.g., $7\rightarrow 6$ via interdiction plan 2), then at least the first batch of supplies would reach demand node 3. 
\item We point out that because supply and demand nodes are closer at the top of the region, the ``throughput'', i.e., demand satisfied per unit time is potentially much higher. Because of this, interdicting Train 2 all the time is not necessarily always a good idea, since Train 2 has to take a long trip to province 12 and finally 4 just for a single batch. 
\item Finally, we remark that Trucks may also operate independently of Train 1. Indeed, if this was not the case, then one would expect Blue to be very brittle. One example of this is seen in Logistics plan 1, Truck 1. Here, Truck 1 moves from province 4 to 7, collects supplies, and transports packages to 3. However, instead of moving directly from 7 to 3, by the path $7\rightarrow 5\rightarrow 8\rightarrow 3$, it takes an detour of $7\rightarrow 15\rightarrow 11\rightarrow 14\rightarrow 3$. This avoids overlapping with Train 1's right path, which contains the segment $7\rightarrow 5$. 
\end{itemize}

\subsubsection{Contested Logistics in Ukraine}
This scenario was constructed based on Ukraine's attack on a key rail bridge connecting the occupied city of Mariupol with Russia on January 7th, 2024. The destruction of this bridge not only disrupts immediate logistical operations but also poses significant long-term challenges to Russia's ability to sustain its military presence in southern Ukraine. Our experiment simulated this scenario using realistic geographical data combined with hypothetical logistics settings to assess the broader impact.

According to the news report~\cite{newsweek2024}, the cities this railway passes through include Mariupol, Donetsk, Taganrog, and Rostov-on-Don, so we defined the area of interest as a rectangle containing these cities. We constructed a realistic map within this area using the Open Street Map (OSM) database, which provides global geographical data, including information on roads, railways, and airports. We used the QGIS software to process the OSM data, extracting and visualizing the nodes and edges to construct accurate transportation networks. The resulting physical graph, with 17 nodes, is depicted on the right in Figure~\ref{fig:ukraine}.

\begin{figure}[t]
    \centering
\begin{subfigure}[t]{0.49\linewidth}
    \centering
    \includegraphics[width=.9\linewidth]{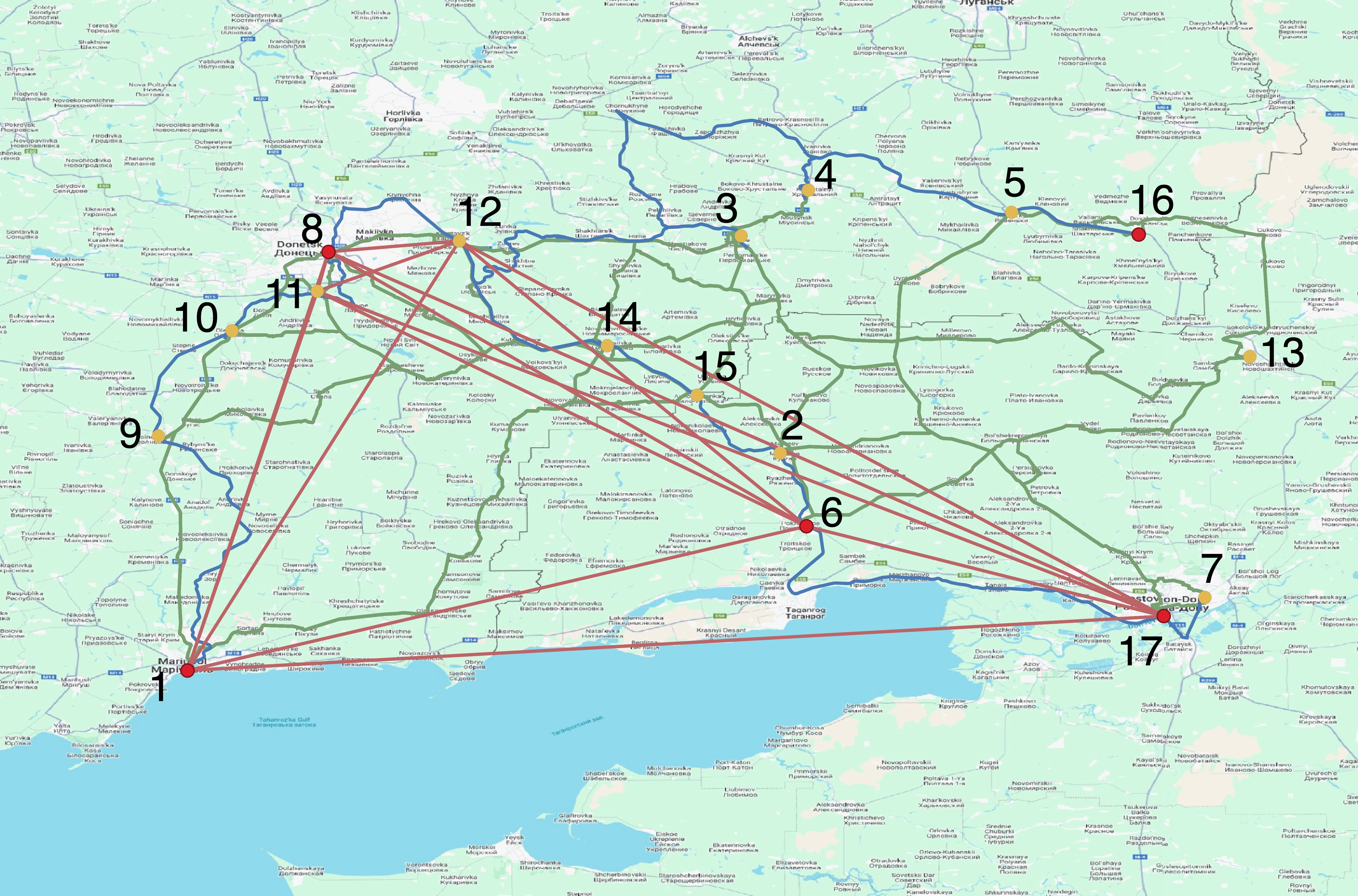}
\end{subfigure}
\hfill
\begin{subfigure}[t]{0.49\linewidth}
    \centering
    \includegraphics[width=.9\linewidth]{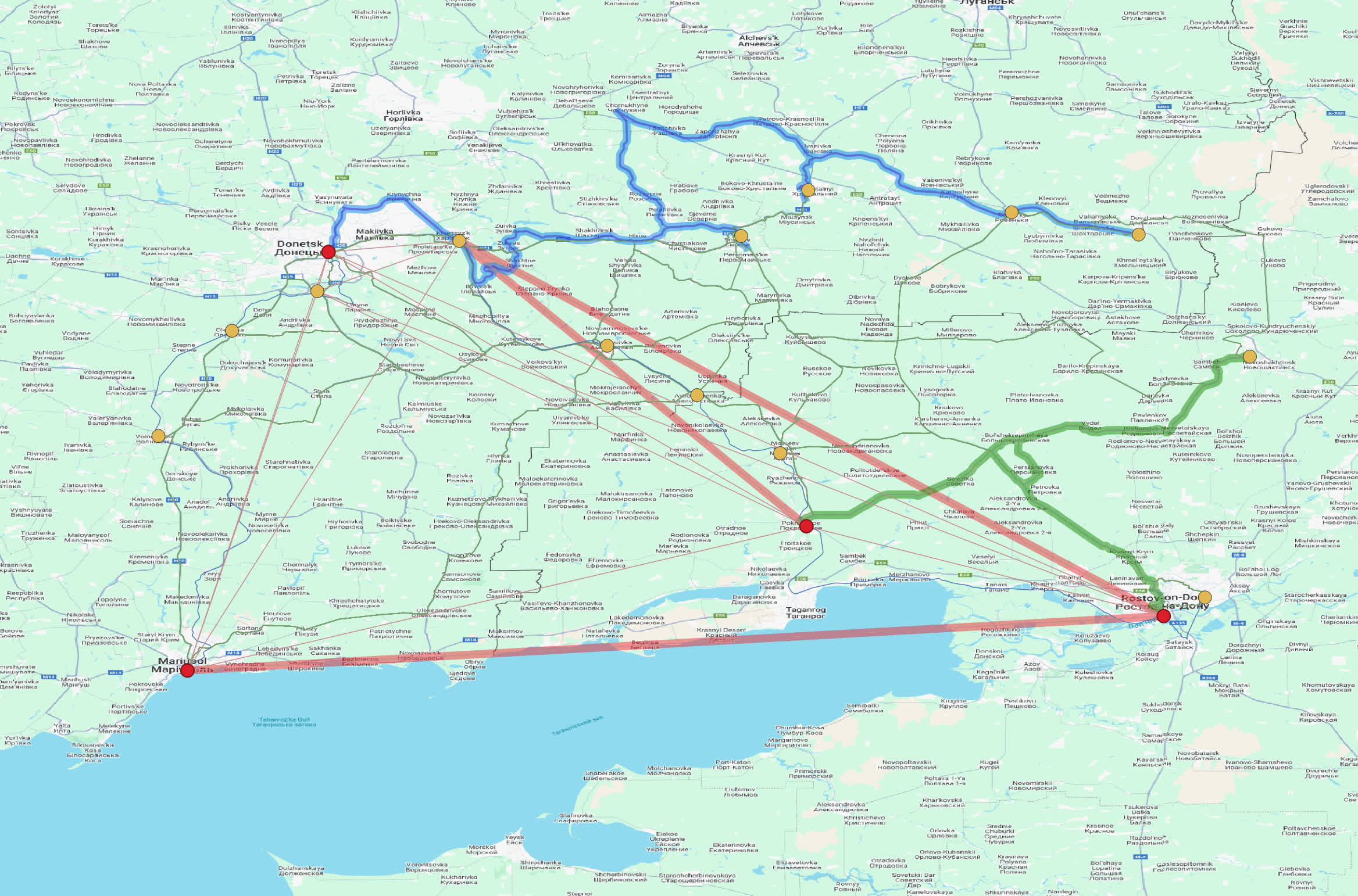}
\end{subfigure}
    \caption{The physical graph for the Ukrainian scenario. The train edges are depicted in blue, the truck edges in green, and the plane edges in red. On the left, the optimal logistics plan for 5 time steps and no Red.}
    \label{fig:ukraine}
\end{figure}

Ukraine and Russia are designated as Red and Blue, respectively. Blue employs three types of connectors: trains, trucks, and planes. To identify the train-accessible edges, we extracted nodes marked as \enquote{train stations} in the OSM data for the area of interest and selected 17 key stations. These included one station each in Mariupol (node 1), Donetsk (node 8), Taganrog (node 6), and Rostov-on-Don (node 17). These stations served as nodes in the train graph. The truck graph used the same nodes, assuming trucks could travel between any train stations. For the plane graph, we selected nodes tagged as \enquote{aerodrome} in the OSM data, representing airports, heliports, and airfields. The plane graph had fewer nodes due to the limited number of airports but included airports in the four major cities. For illustration, each node in the plane graph was considered the same as the closest node in the train and truck graphs, depicted as a single node in Figure~\ref{fig:ukraine}. We then defined edges for the connector graphs. For the train graph, we filtered railway paths in QGIS to connect the 17 stations, depicted in blue in Figure~\ref{fig:ukraine}. For the truck graph, we selected roads tagged as primary, secondary, or tertiary under \enquote{highway} in QGIS and identified the shortest paths connecting the 17 stations, depicted in green in Figure~\ref{fig:ukraine}. For the plane graph, we assumed direct flights between airports, defining edges as direct lines between them. We used QGIS to record the distance of each edge, essential for calculating the time for a connector to traverse an edge given its speed.

\begin{figure}[t]
    \centering
\begin{subfigure}[t]{0.19\linewidth}
    \centering
    \includegraphics[width=\linewidth]{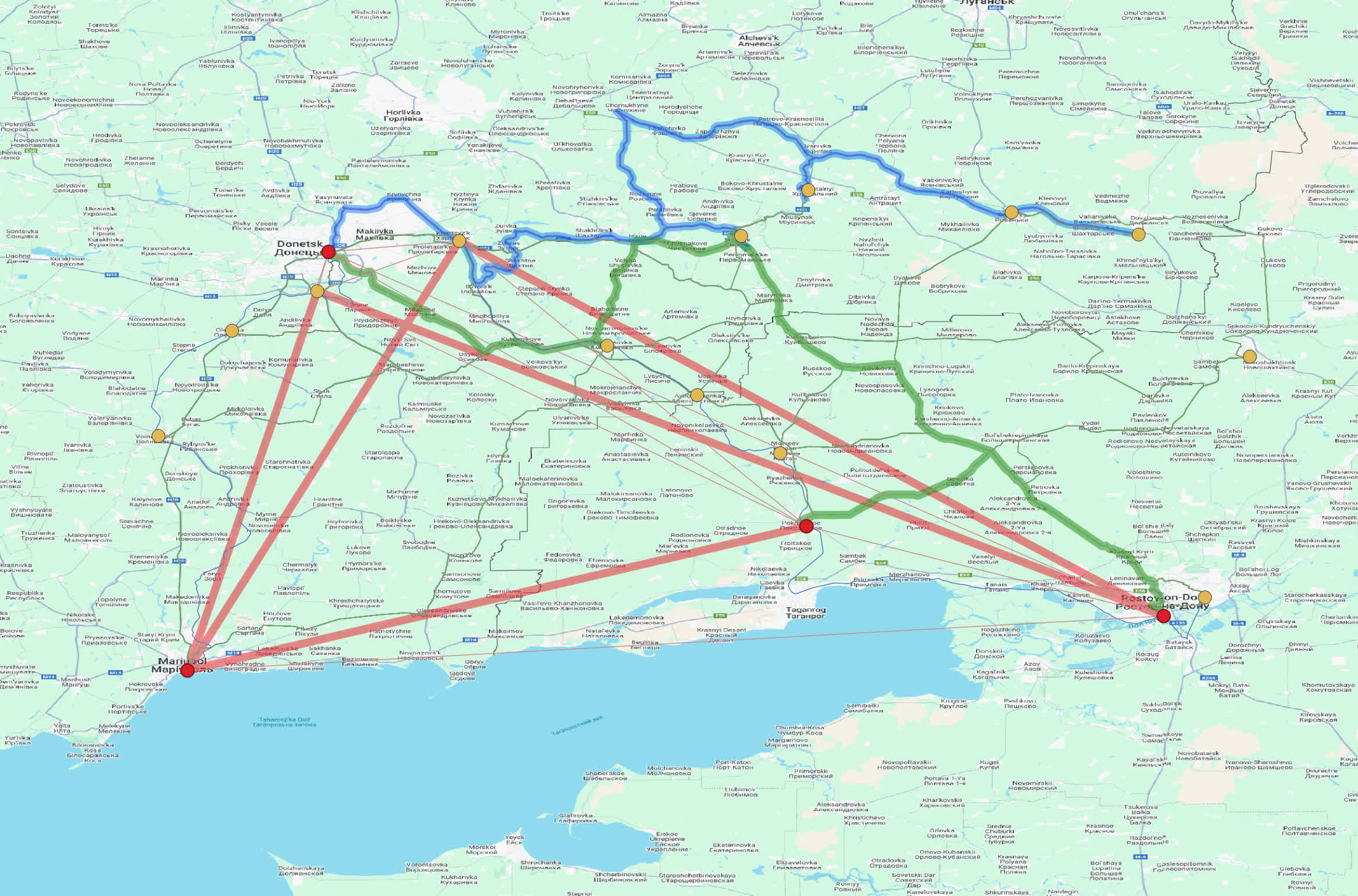}
\end{subfigure}
\hfill
\begin{subfigure}[t]{0.19\linewidth}
    \centering
    \includegraphics[width=\linewidth]{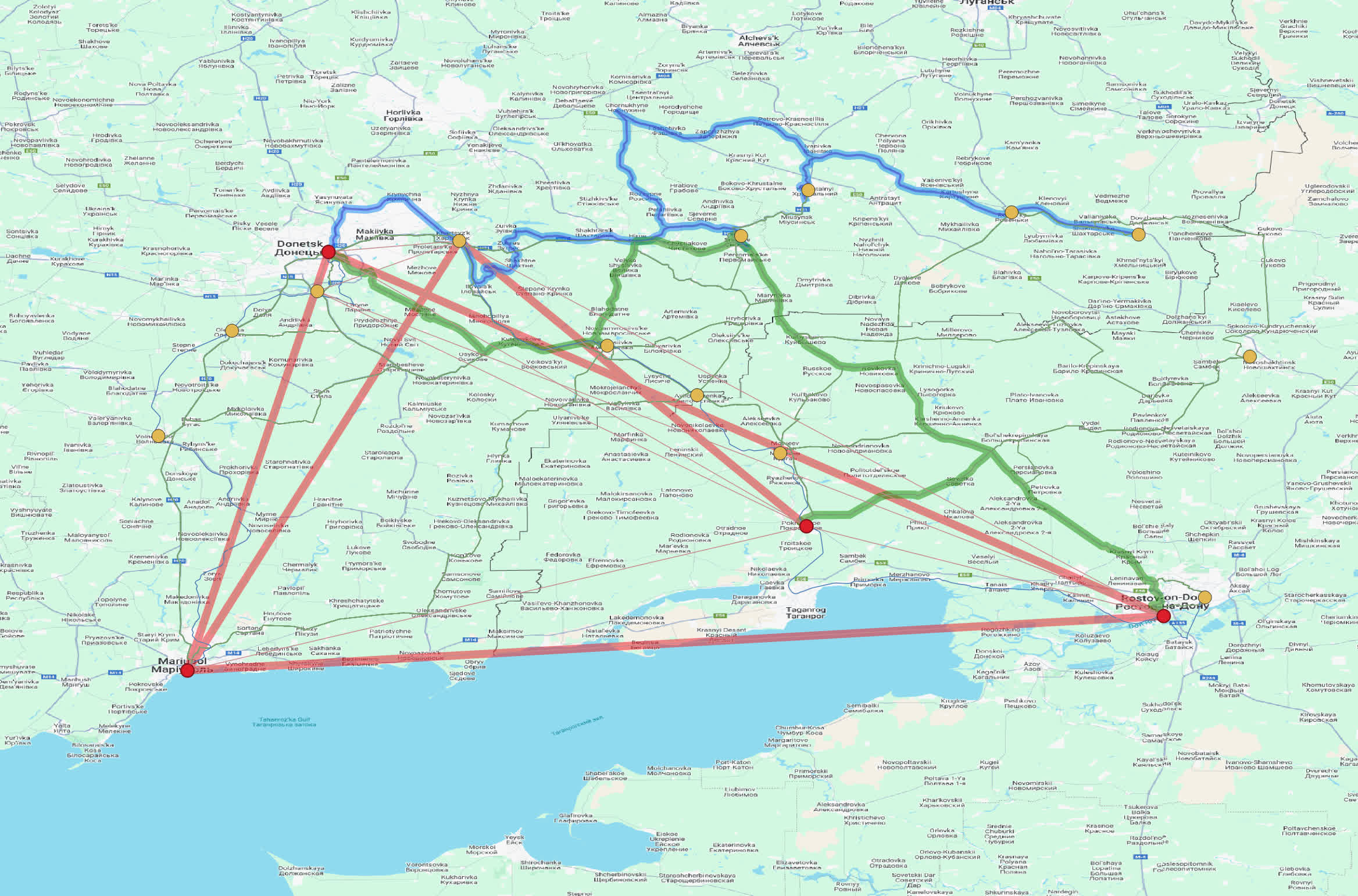}
\end{subfigure}
\hfill
\begin{subfigure}[t]{0.19\linewidth}
    \centering
    \includegraphics[width=\linewidth]{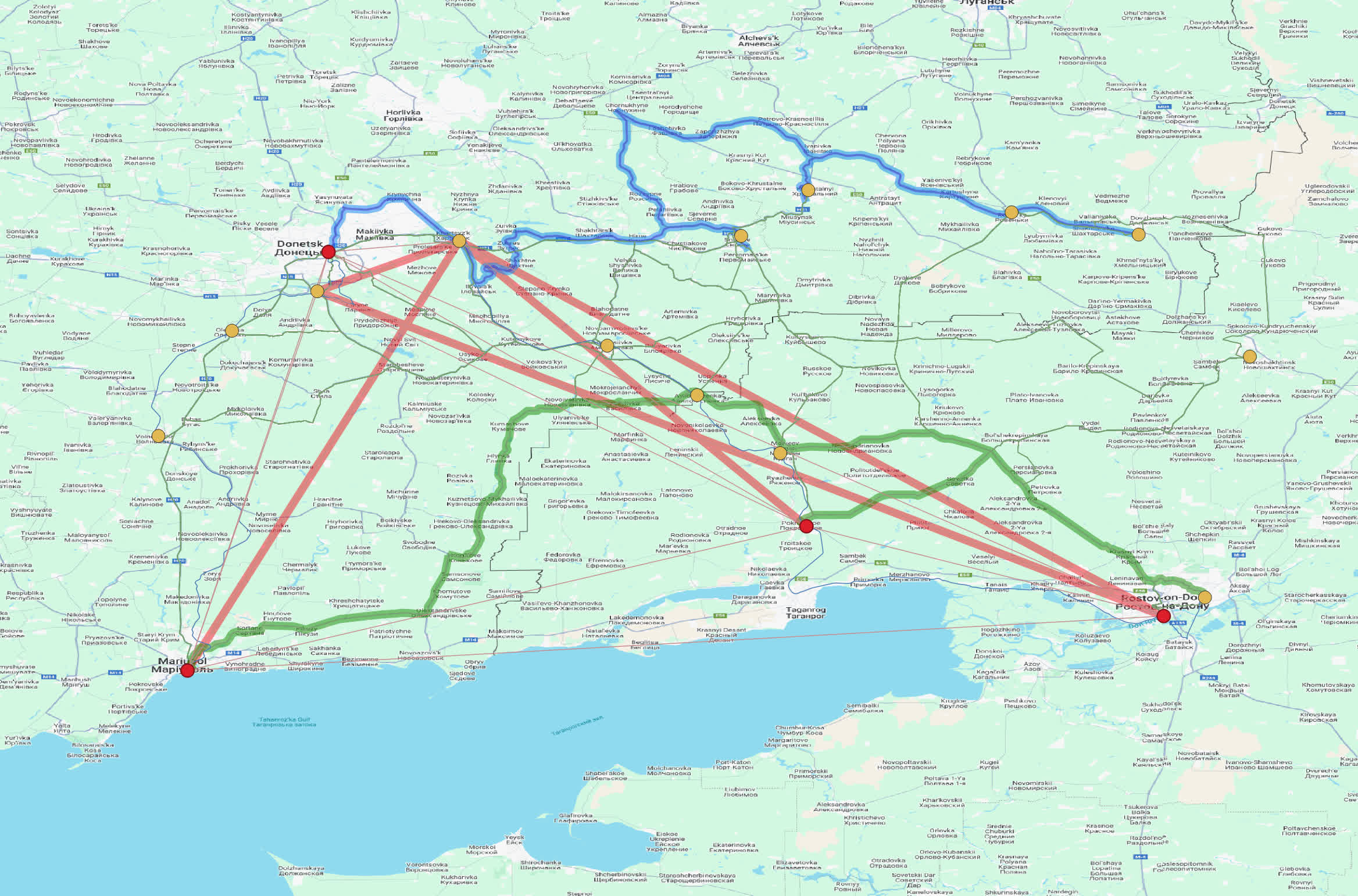}
\end{subfigure}
\hfill
\begin{subfigure}[t]{0.19\linewidth}
    \centering
    \includegraphics[width=\linewidth]{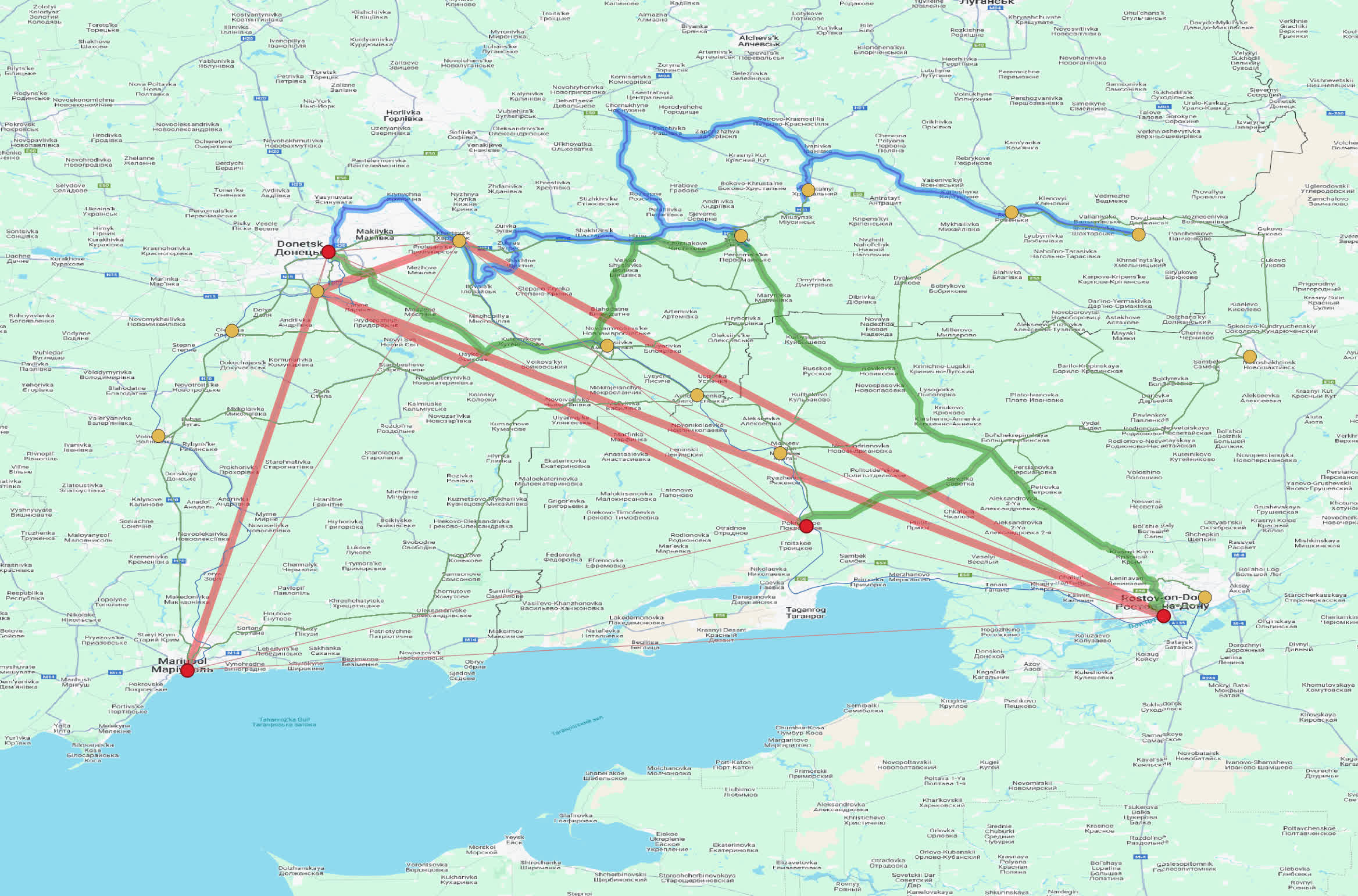}
\end{subfigure}
\hfill
\begin{subfigure}[t]{0.19\linewidth}
    \centering
    \includegraphics[width=\linewidth]{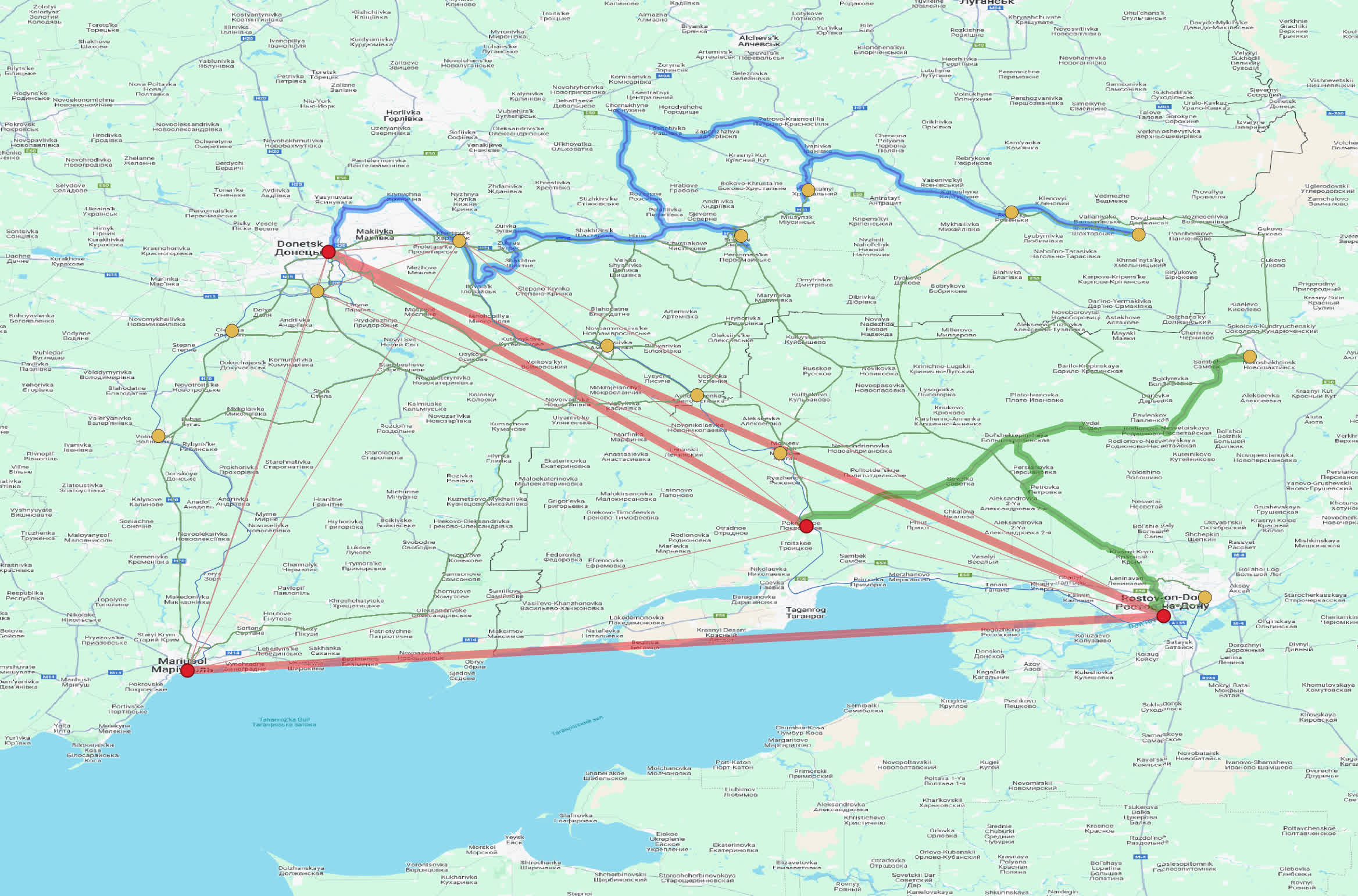}
\end{subfigure}
    \caption{NE for Blue in the Ukraine scenario. Each subfigure corresponds to a single logistics plan played with probabilities 0.339, 0.337, 0.093, 0.076 and 0.062, respectively. Train, Truck and Plane routes are depicted in blue, green and red.}
    \label{fig:ukraine:blue}
\end{figure}

\begin{figure}[t]
    \centering
\begin{subfigure}[t]{0.19\linewidth}
    \centering
    \includegraphics[width=\linewidth]{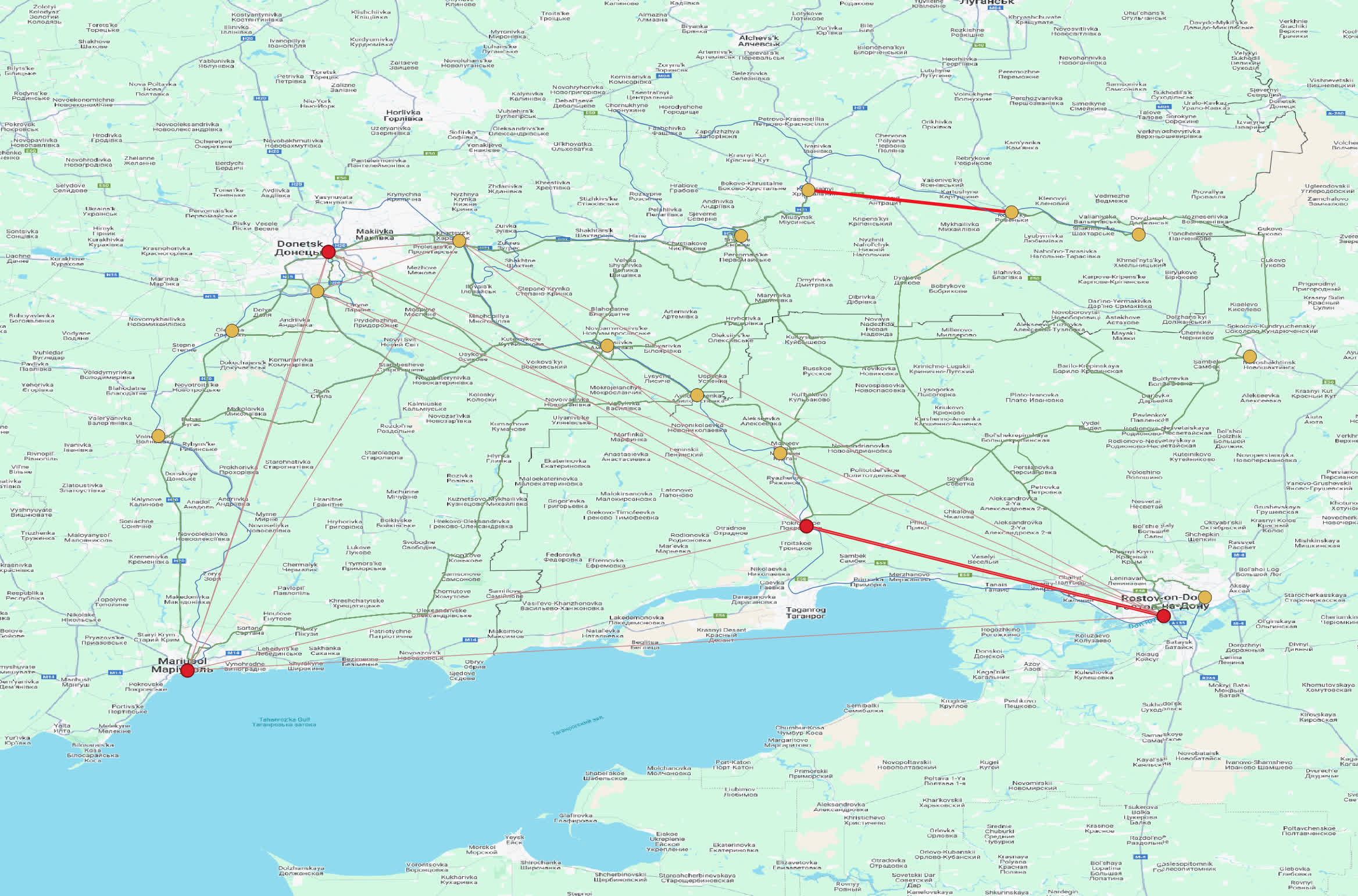}
\end{subfigure}
\hfill
\begin{subfigure}[t]{0.19\linewidth}
    \centering
    \includegraphics[width=\linewidth]{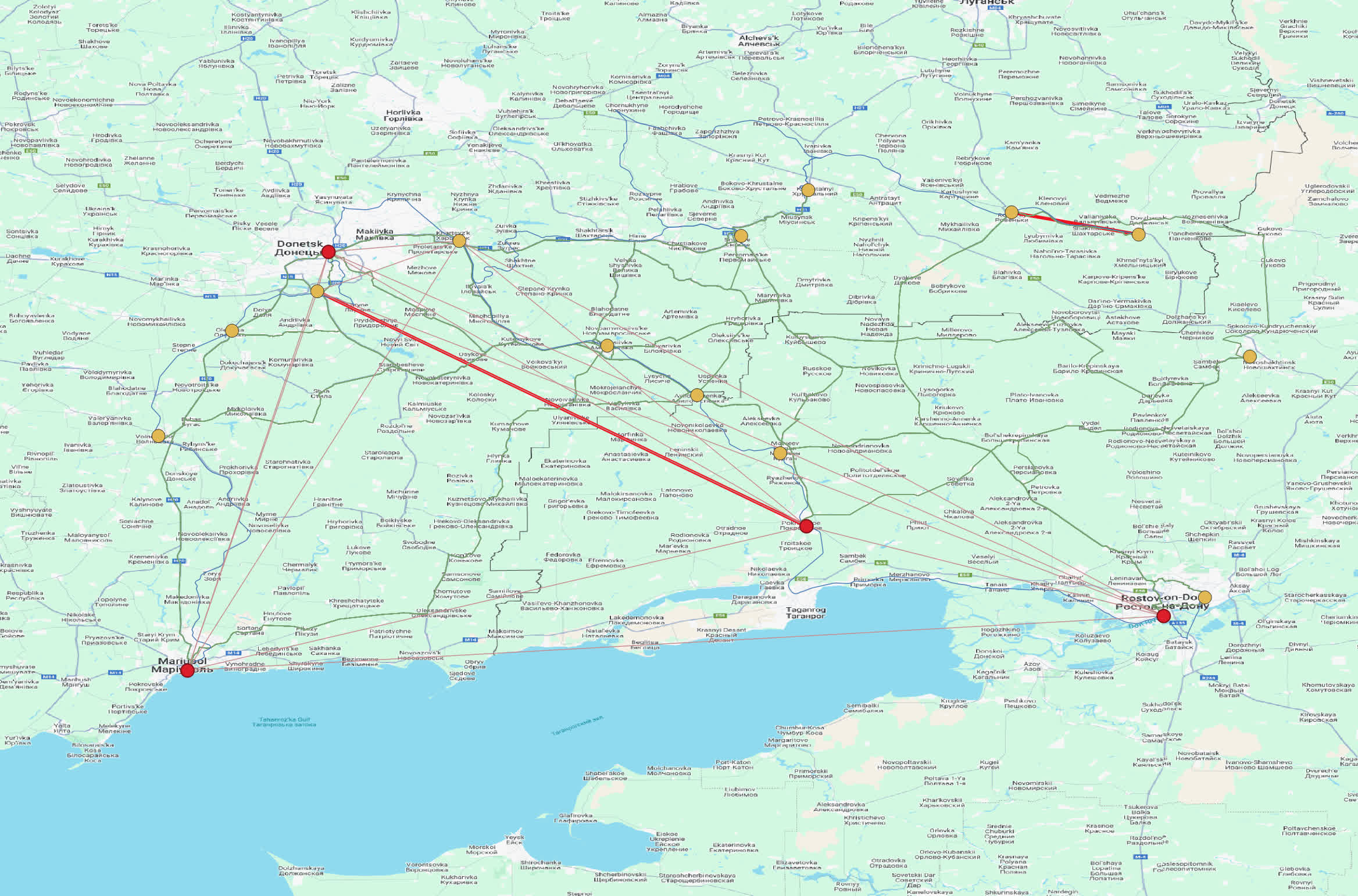}
\end{subfigure}
\hfill
\begin{subfigure}[t]{0.19\linewidth}
    \centering
    \includegraphics[width=\linewidth]{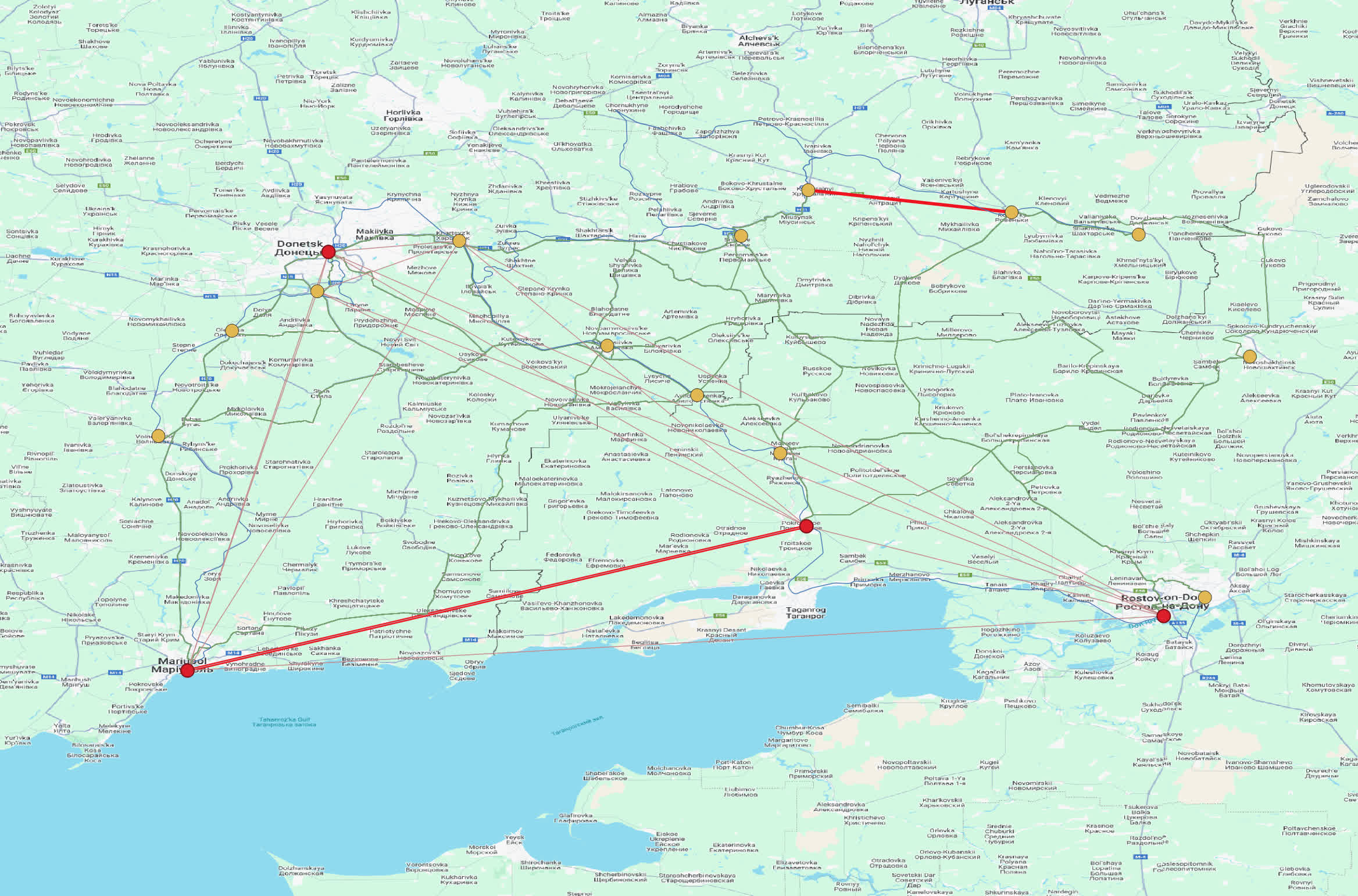}
\end{subfigure}
\hfill
\begin{subfigure}[t]{0.19\linewidth}
    \centering
    \includegraphics[width=\linewidth]{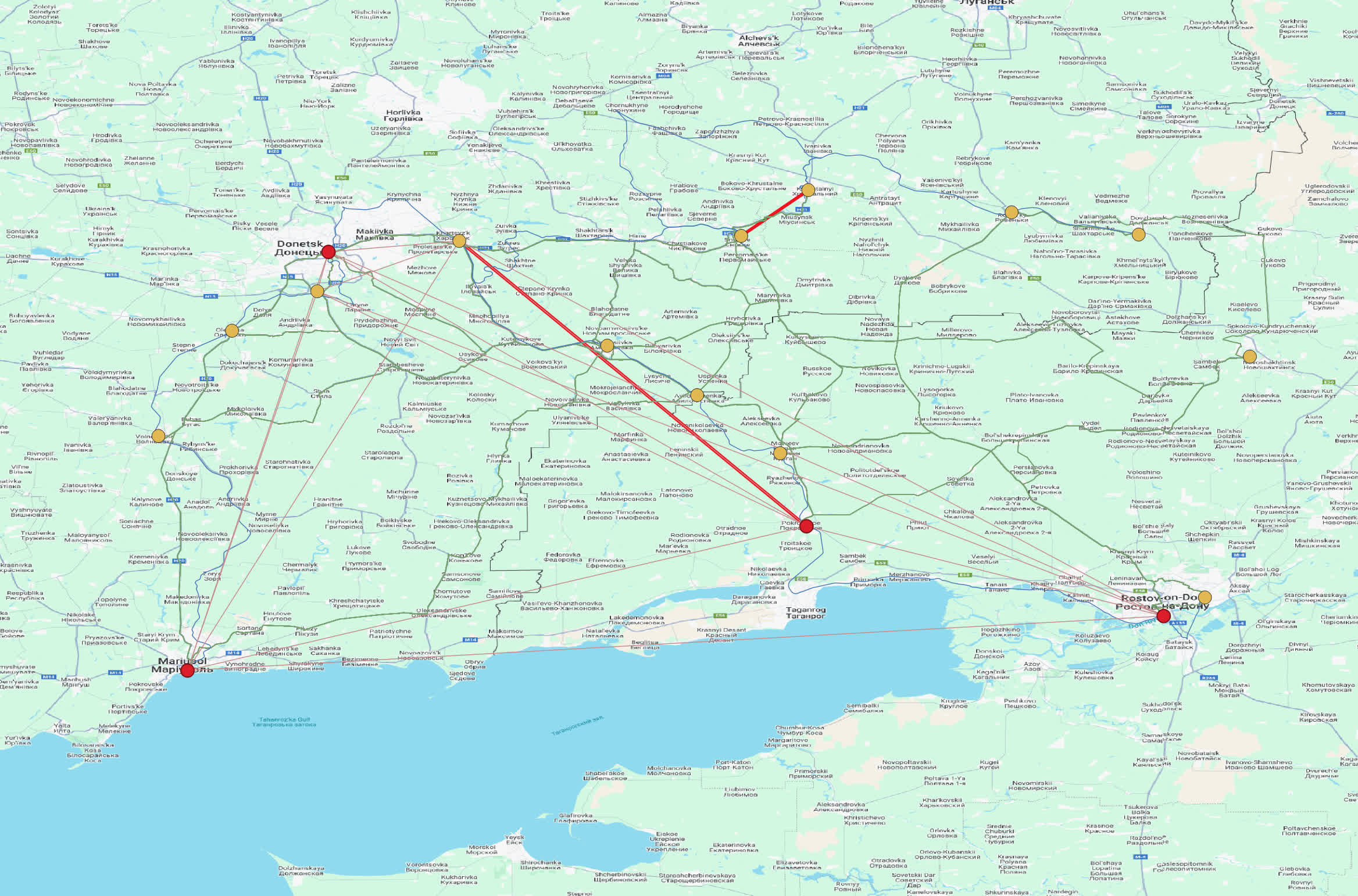}
\end{subfigure}
\hfill
\begin{subfigure}[t]{0.19\linewidth}
    \centering
    \includegraphics[width=\linewidth]{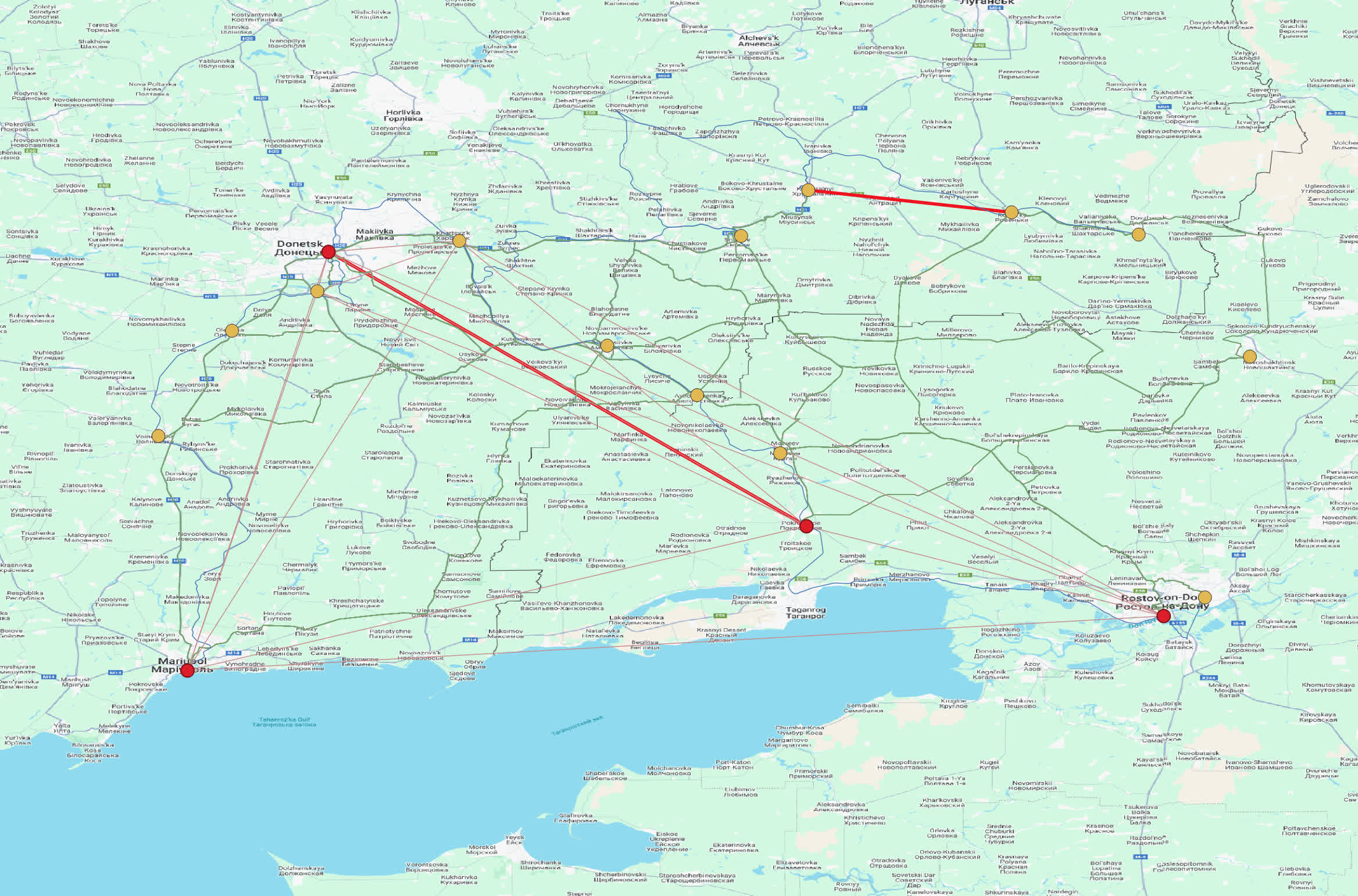}
\end{subfigure}
    \caption{NE for Red in the Ukraine scenario. Each subfigure corresponds to a single interdiction plan played with probabilities 0.462, 0.144, 0.131, 0.131 and 0.131, respectively.}
    \label{fig:ukraine:red}
\end{figure}

In the experiments, Blue has one connector of each type: a Train, a Truck, and a Plane, with each time step representing one hour. The Train starts at node 16, moving at 200 km/h, allowing it to traverse any edge in one time step. The Truck starts at node 6, moving at 100 km/h, taking one time step to cross most edges and two time steps for eight specific edges. The Plane also starts at node 6, moving at 300 km/h, reaching all destinations in one time step. All connectors have sufficient capacity for transferring available loads.

The railway-connected cities (nodes 1, 6, 8, 17) were considered warehouses. With the conflict pushing from the southeast to the northwest, demand locations were set in the northwest (nodes 1 and 8), and supply locations in the southeast (nodes 6, 7, and 17). Additional supply nodes included 3, 13, and 16, strategically chosen for their proximity to the Russian mainland. Supply nodes had packages of types A and B, with unit weights and volumes. Major supply nodes 6 and 17 had 5 units each of A and B, while nodes 3 and 7 had 3 units of A and 1 of B, node 13 had 2 units of both A and B, and node 16 had 3 units of both. The demand was 11 units of A and 13 units of B at node 1, and 14 units of A and 11 units of B at node 8. The payoff per unit was 1.3 at node 1 and 1.1 at node 8.

Consider Blue's optimal logistics plan for a 5-step horizon, assuming no presence of Red, as shown on the right side of Figure~\ref{fig:ukraine}. The Train's route is marked in blue, the Truck's in green, and the Plane's in red. The Plane delivers supplies to node 12, which the Train then transports to node 8. Meanwhile, the Truck collects packages from eastern warehouses and brings them to node 17, where they are loaded onto the Plane and flown to node 1. The value of this plan is 1.5956, but it is highly exploitable. Red can easily reduce the payoff to 0 by interdicting any two critical edges and cutting off deliveries to nodes 1 and 8.

To address this vulnerability, the equilibrium of the corresponding CL game, where Red can interdict any two edges, introduces randomization. Figure~\ref{fig:ukraine:blue} depicts Blue's mixed strategy, showing the probabilities of playing each of the five logistics plans. The paths are more randomized, and even the connectors responsible for final deliveries to demand nodes may change. For example, the Truck only follows its original route from the no-Red scenario in the least frequently played plan. More often, it delivers supplies directly to nodes 8 or 1. Similarly, Figure~\ref{fig:ukraine:red} shows Red's randomized interdiction strategy. As expected, Red consistently targets the Train, which is vulnerable due to movement constraints, and frequently interdicts the route between the region’s key cities, which serve as major transport hubs. In other cases, Red attempts to intercept the Plane. The value of this equilibrium is 1.0401, about 65\% of the optimal no-Red logistics value, but it is significantly more robust against adversarial actions.

\begin{table}[t]
    \centering
    \begin{minipage}{0.45 \linewidth}
    \centering
    \begin{tabular}{c|p{.18\linewidth}p{.18\linewidth}p{.18\linewidth}}
    Exp.$\backslash$True & 0 & 1 & 2 \\\hline
         0 & \textbf{31.7} & 8.33 & 0.0 \\
         1 & 25.8 & \textbf{16.7} & 6.67 \\
         2 & 23.3 & 15.9 & \textbf{9.26} \\
    \end{tabular}
    \end{minipage}
    \begin{minipage}{0.45 \linewidth}
    \centering
    \begin{tabular}{c|p{.18\linewidth}p{.18\linewidth}p{.18\linewidth}p{.18\linewidth}}
    Exp.$\backslash$True & 0 & 1 & 2 & 3 \\\hline
         0 & \textbf{1.60} & 0.40 & 0.0 & 0.0 \\
         1 & 1.60 & \textbf{1.20} & 0.78 & 0.38 \\
         2 & 1.44 & 1.12 & \textbf{1.04} & 0.64 \\
         3 & 1.32 & 1.10 & 0.99 & \textbf{0.89}
    \end{tabular}
    \end{minipage} \quad 
    \caption{The performance of the computed Blue's randomized logistics plans against less or more capable Red in (left) the UK scenario and (right) the Ukraine scenario. Rows correspond to the \textit{expected} budget used during the computation. Columns indicate the \textit{true} Red's budget.} 
    \label{tab:price_robustness}
\end{table}
\subsubsection{Price of Robustness}
It is natural to ask: what if Blue does not know Red's budget? What happens if Blue assumes Red has a high interdiction budget when, in fact, it has none, or vice versa? This leads to a discussion about the \textit{Price of Robustness}, which refers to the amount Blue sacrifices to be robust against Red. We present our findings in Table~\ref{tab:price_robustness} for both the UK and Mariupol scenarios. It is evident in both cases that underestimating Red's capabilities results in a significant drop in utility. For example, when Blue devises a logistics plan without considering Red (i.e., the first row of each table), its utility drops to zero with just a Red budget of 2. Conversely, adopting a more conservative logistics plan (i.e., assuming Red has a higher budget) leads to a relatively smaller drop in utility. For instance, in the UK scenario, the Price of Robustness is just 31.7-23.3, which is approximately one-quarter of the expected utility. A similar trend is observed in the Ukraine scenario.

\subsubsection{Comparisons Against a Non-Game-Theoretic Alternative}
We now compare our game-theoretic approach to a simple heuristic that does not explicitly account for Red. Consider the \textit{min-overlap} heuristic, which involves two hyperparameters: $k$, the minimum target payoff, and \# str, the number of logistics plans played with positive probability. The min-overlap heuristic identifies \# str logistics plans, each required to achieve a utility of $k$ under the assumption that \textit{Red does not exist}, while minimizing the maximum overlap across edges. Here, overlap on a given physical edge is the total number of connectors using that edge, summed across all \# str logistics plans. Blue then randomizes uniformly over these \# str logistics plans. The rationale is that by minimizing overlap among ``good'' logistics plans, no single edge will be excessively used by connectors across the logistics plans. We tested this min-overlap heuristic strategy against a best-responding Red and found that it performs poorly, as shown in Table~\ref{tab:heuristic}. For instance, in the UK scenario, none of the instances achieve more than half of the true Nash value (9.259). In the Ukraine scenario, the best-performing instance reached only 69\% of the Nash value. This suggests that seemingly reasonable heuristics may actually perform poorly in practice, and that counter-intuitive logistics plans may be necessary for optimal performance.

\begin{table}[t]
    \centering
    \begin{minipage}{0.45\linewidth}
    \begin{tabular}{c|ccccc}
    $k$ $\backslash$ \# str & 3 & 4 & 5 & 6 & 7 \\\hline
         10 & 1.11 & 1.67 & 1.33 & 1.94 & 1.67 \\
         20 & 3.89 & 2.92 & 3.33 & 2.50 & 4.29 \\
         30 & 2.22 & 3.33 & 2.33 & 2.22 & 3.10 \\
    \end{tabular}
    \end{minipage}
    \begin{minipage}{0.45\linewidth}
    \begin{tabular}{c|ccccc}
    $k$ $\backslash$ \# str & 5 & 10 & 15 & 20 & 25  \\\hline
         1.2 & 0.720 & 0.410 & 0.627 & 0.730 & 0.641  \\
         1.3 & 0.517 & 0.496 & 0.570 & 0.676 & 0.577 \\
         1.4 & 0.439 & 0.588 & 0.553 & 0.584 & 0.648  \\
         1.5 & 0.517 & 0.496 & 0.344 & 0.576 & 0.637 
    \end{tabular} 
    \end{minipage}
    \caption{Exploitability of the min-overlap heuristic strategies in (left) the UK scenario and (right) the Ukraine scenario.}
    \label{tab:heuristic}
\end{table}

%% file: figs/grid_5x5_uni.tex
\begin{tikzpicture}[scale=\fscale]
\begin{semilogyaxis}[
    title={Grid world [$5\times 5$]},
    xlabel={Game size [\# of timesteps]},
    ylabel style={align=center},
    ylabel={Runtime [s]},
    legend pos=north west,
    ymajorgrids=true,
    grid style=dashed,
    scaled ticks=false, 
    title style={font=\titfont},
    label style={font=\labfont},
    tick label style={/pgf/number format/fixed, font=\tickfont}  
]

\legend{rb = 1, rb = 2, rb = 3, rb = 4, rb = 5}
\addplot[color=black,mark=o,error bars/.cd,y dir=both, y explicit]
coordinates {
(6, 71.85600000000001) += (0, 0.10901569177247618) -= (0, 0.10901569177247618) 
(7, 19.079999999999995) += (0, 1.1678944166409275) -= (0, 1.1678944166409275) 
(8, 21.5465) += (0, 1.8819568867761585) -= (0, 1.8819568867761585) 
(9, 36.2665) += (0, 4.339211138286725) -= (0, 4.339211138286725) 
(10, 53.1985) += (0, 8.318342138564061) -= (0, 8.318342138564061) 
(11, 354.861) += (0, 126.0637723157358) -= (0, 126.0637723157358) 
(12, 351.6555) += (0, 120.2587303111066) -= (0, 120.2587303111066) 
(13, 2696.8519999999994) += (0, 1654.629104331971) -= (0, 1654.629104331971) 
(14, 3367.478) += (0, 1380.0699621397155) -= (0, 1380.0699621397155) 
};

\addplot[color=black,mark=square,error bars/.cd,y dir=both, y explicit]
coordinates {
(6, 19.25117647058824) += (0, 2.2120425552949894) -= (0, 2.2120425552949894) 
(7, 29.9375) += (0, 3.562995404150839) -= (0, 3.562995404150839) 
(8, 38.5925) += (0, 5.337866115879325) -= (0, 5.337866115879325) 
(9, 92.995) += (0, 11.48093466532554) -= (0, 11.48093466532554) 
(10, 98.25649999999999) += (0, 11.364481586226631) -= (0, 11.364481586226631) 
(11, 3675.907000000001) += (0, 2456.4944659510115) -= (0, 2456.4944659510115) 
(12, 2305.017368421052) += (0, 1237.1044012091845) -= (0, 1237.1044012091845) 
(13, 1916.9762500000006) += (0, 1288.0761478201628) -= (0, 1288.0761478201628) 
(14, 988.8410526315787) += (0, 295.8955989811772) -= (0, 295.8955989811772) 
};

\addplot[color=black,mark=pentagon,error bars/.cd,y dir=both, y explicit]
coordinates {
(6, 13.206999999999999) += (0, 0.4039372671461173) -= (0, 0.4039372671461173) 
(7, 29.482499999999998) += (0, 3.42603890313427) -= (0, 3.42603890313427) 
(8, 40.55049999999999) += (0, 9.310129471511972) -= (0, 9.310129471511972) 
(9, 126.71849999999999) += (0, 26.764804651621755) -= (0, 26.764804651621755) 
(10, 159.1275) += (0, 49.06492308465011) -= (0, 49.06492308465011) 
(11, 489.65352941176474) += (0, 179.20027898879684) -= (0, 179.20027898879684) 
(12, 4021.9194117647053) += (0, 2572.1921123957545) -= (0, 2572.1921123957545) 
(13, 3319.7925) += (0, 1998.173163467592) -= (0, 1998.173163467592) 
};

\end{semilogyaxis}
\end{tikzpicture}

%% file: figs/grid_6x6_uni.tex
\begin{tikzpicture}[scale=\fscale]
\begin{semilogyaxis}[
    title={Grid world [$6\times 6$]},
    xlabel={Game size [\# of timesteps]},
    ylabel style={align=center},
    ylabel={Runtime [s]},
    legend pos=north west,
    ymajorgrids=true,
    grid style=dashed,
    scaled ticks=false, 
    title style={font=\titfont},
    label style={font=\labfont},
    tick label style={/pgf/number format/fixed, font=\tickfont}  
]

\legend{rb = 1, rb = 2, rb = 3, rb = 4, rb = 5}
\addplot[color=black,mark=o,error bars/.cd,y dir=both, y explicit]
coordinates {
(6, 34.3875) += (0, 0.19322045631051094) -= (0, 0.19322045631051094) 
(7, 15.228500000000002) += (0, 1.4891064567295662) -= (0, 1.4891064567295662) 
(8, 22.765500000000003) += (0, 1.4531724777190078) -= (0, 1.4531724777190078) 
(9, 26.9275) += (0, 2.7508188135056884) -= (0, 2.7508188135056884) 
(10, 127.625) += (0, 23.528244610970965) -= (0, 23.528244610970965) 
(11, 152.651) += (0, 20.712919168327872) -= (0, 20.712919168327872) 
(12, 502.1385) += (0, 161.6628055671092) -= (0, 161.6628055671092) 
(13, 503.247) += (0, 140.92105462596686) -= (0, 140.92105462596686) 
(14, 1264.614) += (0, 341.234036171349) -= (0, 341.234036171349) 
};

\addplot[color=black,mark=square,error bars/.cd,y dir=both, y explicit]
coordinates {
(6, 12.220499999999998) += (0, 0.22547782596078042) -= (0, 0.22547782596078042) 
(7, 13.057500000000001) += (0, 0.8764476370704161) -= (0, 0.8764476370704161) 
(8, 41.910000000000004) += (0, 8.503000244495746) -= (0, 8.503000244495746) 
(9, 45.570499999999996) += (0, 6.70178568014761) -= (0, 6.70178568014761) 
(10, 168.753) += (0, 26.489526163531522) -= (0, 26.489526163531522) 
(11, 205.85650000000004) += (0, 45.399323628598374) -= (0, 45.399323628598374) 
(12, 1724.5194999999999) += (0, 1162.4487903488084) -= (0, 1162.4487903488084) 
(13, 1103.4445) += (0, 227.73952545729333) -= (0, 227.73952545729333) 
(14, 2281.8709999999996) += (0, 498.0095971210636) -= (0, 498.0095971210636) 
};

\addplot[color=black,mark=pentagon,error bars/.cd,y dir=both, y explicit]
coordinates {
(6, 11.682500000000001) += (0, 0.11075310379397947) -= (0, 0.11075310379397947) 
(7, 11.624500000000001) += (0, 0.1645031994569678) -= (0, 0.1645031994569678) 
(8, 40.688500000000005) += (0, 6.4856976768319345) -= (0, 6.4856976768319345) 
(9, 47.989) += (0, 9.902633485460743) -= (0, 9.902633485460743) 
(10, 188.4875) += (0, 52.777447799699445) -= (0, 52.777447799699445) 
(11, 217.121) += (0, 52.35515964569601) -= (0, 52.35515964569601) 
(12, 6303.524210526316) += (0, 2479.7461993397524) -= (0, 2479.7461993397524) 
(13, 2518.67) += (0, 1454.9274876753595) -= (0, 1454.9274876753595) 
(14, 1540.40125) += (0, 565.3617601057714) -= (0, 565.3617601057714) 
};

\end{semilogyaxis}
\end{tikzpicture}

%% file: figs/grid_7x7_uni.tex
\begin{tikzpicture}[scale=\fscale]
\begin{semilogyaxis}[
    title={Grid world [$7\times 7$]},
    xlabel={Game size [\# of timesteps]},
    ylabel style={align=center},
    ylabel={Runtime [s]},
    legend pos=north west,
    ymajorgrids=true,
    grid style=dashed,
    scaled ticks=false, 
    title style={font=\titfont},
    label style={font=\labfont},
    tick label style={/pgf/number format/fixed, font=\tickfont}  
]

\legend{rb = 1, rb = 2, rb = 3, rb = 4, rb = 5}
\addplot[color=black,mark=o,error bars/.cd,y dir=both, y explicit]
coordinates {
(6, 16.8005) += (0, 4.219375887685162) -= (0, 4.219375887685162) 
(7, 13.173684210526316) += (0, 0.20190497812294758) -= (0, 0.20190497812294758) 
(8, 12.108) += (0, 0.1380953141474166) -= (0, 0.1380953141474166) 
(9, 22.351052631578955) += (0, 1.4819488699217933) -= (0, 1.4819488699217933) 
(10, 30.139999999999997) += (0, 2.588103247757444) -= (0, 2.588103247757444) 
(11, 92.14368421052633) += (0, 9.756589442526284) -= (0, 9.756589442526284) 
(12, 108.8885) += (0, 11.01157401328257) -= (0, 11.01157401328257) 
(13, 227.14315789473685) += (0, 54.16643870425784) -= (0, 54.16643870425784) 
(14, 305.561052631579) += (0, 67.92066303986364) -= (0, 67.92066303986364) 
};

\addplot[color=black,mark=square,error bars/.cd,y dir=both, y explicit]
coordinates {
(6, 10.936000000000002) += (0, 0.05528109984434101) -= (0, 0.05528109984434101) 
(7, 11.6625) += (0, 0.1081297344567359) -= (0, 0.1081297344567359) 
(8, 12.946000000000002) += (0, 0.3121195789740645) -= (0, 0.3121195789740645) 
(9, 40.152631578947364) += (0, 3.6840872326905147) -= (0, 3.6840872326905147) 
(10, 55.4905) += (0, 6.339572269939785) -= (0, 6.339572269939785) 
(11, 171.68526315789472) += (0, 21.032507384572995) -= (0, 21.032507384572995) 
(12, 226.16578947368419) += (0, 32.2373092607278) -= (0, 32.2373092607278) 
(13, 481.2358823529412) += (0, 106.34149987138697) -= (0, 106.34149987138697) 
(14, 725.1005555555556) += (0, 148.13133807288347) -= (0, 148.13133807288347) 
};

\addplot[color=black,mark=pentagon,error bars/.cd,y dir=both, y explicit]
coordinates {
(6, 10.738421052631578) += (0, 0.03576452518805114) -= (0, 0.03576452518805114) 
(7, 11.590526315789473) += (0, 0.1211964566391175) -= (0, 0.1211964566391175) 
(8, 11.766) += (0, 0.13475591970201845) -= (0, 0.13475591970201845) 
(9, 42.0378947368421) += (0, 5.294483092404648) -= (0, 5.294483092404648) 
(10, 49.89105263157896) += (0, 6.724813347471112) -= (0, 6.724813347471112) 
(11, 206.87736842105264) += (0, 46.07642946186745) -= (0, 46.07642946186745) 
(12, 266.8836842105263) += (0, 71.34312115670595) -= (0, 71.34312115670595) 
(13, 484.5552631578947) += (0, 108.39810908875033) -= (0, 108.39810908875033) 
(14, 589.8294736842106) += (0, 146.18972343998155) -= (0, 146.18972343998155) 
};

\end{semilogyaxis}
\end{tikzpicture}

%% file: figs/grid_5x5_runtimes.tex
\begin{tikzpicture}[scale=\fscale]
\begin{semilogyaxis}[
    title={Grid world [$5\times 5$]},
    xlabel={Game size [\# of timesteps]},
    ylabel style={align=center},
    ylabel={Runtime [s]},
    legend pos=north west,
    ymajorgrids=true,
    grid style=dashed,
    scaled ticks=false, 
    title style={font=\titfont},
    label style={font=\labfont},
    tick label style={/pgf/number format/fixed, font=\tickfont}  
]

\legend{rb = 1, rb = 2, rb = 3, rb = 4, rb = 5}
\addplot[color=black,mark=o,error bars/.cd,y dir=both, y explicit]
coordinates {
(6, 38.51233333333333) += (0, 1.2828516800278165) -= (0, 1.2828516800278165)
(7, 28.00866666666667) += (0, 2.339212694615721) -= (0, 2.339212694615721)
(8, 10.914) += (0, 0.7045239038063099) -= (0, 0.7045239038063099)
(9, 14.222758620689653) += (0, 1.1718125088088613) -= (0, 1.1718125088088613)
(10, 14.942333333333329) += (0, 1.2660524841886716) -= (0, 1.2660524841886716)
(11, 35.866) += (0, 16.286152778413907) -= (0, 16.286152778413907)
(12, 32.02833333333333) += (0, 8.389683608673204) -= (0, 8.389683608673204)
(13, 111.67800000000001) += (0, 58.81354850048066) -= (0, 58.81354850048066)
(14, 91.44633333333333) += (0, 42.96792441329897) -= (0, 42.96792441329897)
(15, 308.7616666666666) += (0, 161.5109381670646) -= (0, 161.5109381670646)
(16, 307.29600000000005) += (0, 171.17729190668294) -= (0, 171.17729190668294)
};

\addplot[color=black,mark=square,error bars/.cd,y dir=both, y explicit]
coordinates {
(6, 8.119333333333335) += (0, 0.07123529625789586) -= (0, 0.07123529625789586)
(7, 13.186333333333334) += (0, 1.1191289185784037) -= (0, 1.1191289185784037)
(8, 15.176551724137934) += (0, 1.3705065652109165) -= (0, 1.3705065652109165)
(9, 54.46366666666665) += (0, 14.682512175689554) -= (0, 14.682512175689554)
(10, 61.36133333333333) += (0, 17.627591015115154) -= (0, 17.627591015115154)
(11, 337.9503333333333) += (0, 130.9918335231868) -= (0, 130.9918335231868)
(12, 402.6279999999999) += (0, 133.8696611534269) -= (0, 133.8696611534269)
(13, 2873.228965517242) += (0, 1228.0130217661288) -= (0, 1228.0130217661288)
(14, 2557.742142857144) += (0, 1506.1005390066514) -= (0, 1506.1005390066514)
(15, 1286.2375) += (0, 482.03269909126146) -= (0, 482.03269909126146)
(16, 2614.5340909090905) += (0, 959.1656021322414) -= (0, 959.1656021322414)
};

\addplot[color=black,mark=pentagon,error bars/.cd,y dir=both, y explicit]
coordinates {
(6, 8.352333333333334) += (0, 0.1664039423530925) -= (0, 0.1664039423530925)
(7, 17.28933333333333) += (0, 2.8329628263496374) -= (0, 2.8329628263496374)
(8, 23.646666666666672) += (0, 5.799325829493317) -= (0, 5.799325829493317)
(9, 181.81266666666664) += (0, 79.01751319472747) -= (0, 79.01751319472747)
(10, 159.62966666666668) += (0, 45.88920347187915) -= (0, 45.88920347187915)
(11, 2779.902) += (0, 1411.6918711420697) -= (0, 1411.6918711420697)
(12, 2213.4372413793103) += (0, 1434.6137609259686) -= (0, 1434.6137609259686)
(13, 1448.595) += (0, 571.5229426269781) -= (0, 571.5229426269781)
(14, 3910.698518518518) += (0, 1733.271170640973) -= (0, 1733.271170640973)
(15, 3559.4404000000004) += (0, 1598.5134074535713) -= (0, 1598.5134074535713)
(16, 5163.410769230772) += (0, 1932.0858478240095) -= (0, 1932.0858478240095)
};

\addplot[color=black,mark=diamond,error bars/.cd,y dir=both, y explicit]
coordinates {
(6, 10.162) += (0, 1.1034214709387957) -= (0, 1.1034214709387957)
(7, 31.327999999999996) += (0, 10.224393857198894) -= (0, 10.224393857198894)
(8, 33.828333333333326) += (0, 10.14015245167917) -= (0, 10.14015245167917)
(9, 205.3373333333333) += (0, 62.602592803587626) -= (0, 62.602592803587626)
(10, 304.8496666666667) += (0, 141.6209721883005) -= (0, 141.6209721883005)
(11, 2594.962500000001) += (0, 1498.805220502911) -= (0, 1498.805220502911)
(12, 2768.9178571428574) += (0, 1292.5241051855894) -= (0, 1292.5241051855894)
(13, 3372.540384615385) += (0, 1441.5626775195656) -= (0, 1441.5626775195656)
(14, 3481.6312) += (0, 1693.2894945640887) -= (0, 1693.2894945640887)
(15, 2786.4955999999993) += (0, 811.7469231382155) -= (0, 811.7469231382155)
(16, 3195.38125) += (0, 995.6358397996686) -= (0, 995.6358397996686)
};

\addplot[color=black,mark=triangle,error bars/.cd,y dir=both, y explicit]
coordinates {
(6, 19.05047619047619) += (0, 2.8091010658546876) -= (0, 2.8091010658546876)
(7, 44.66555555555556) += (0, 19.7873645102851) -= (0, 19.7873645102851)
(8, 21.945714285714285) += (0, 2.1877188131572525) -= (0, 2.1877188131572525)
(9, 425.4886206896552) += (0, 173.57416522960358) -= (0, 173.57416522960358)
(10, 399.9220000000001) += (0, 129.241051507656) -= (0, 129.241051507656)
(11, 696.4852000000002) += (0, 310.9488282998346) -= (0, 310.9488282998346)
(12, 1053.5791666666667) += (0, 529.6517478788578) -= (0, 529.6517478788578)
(13, 2550.947826086957) += (0, 1180.8068451745403) -= (0, 1180.8068451745403)
(14, 3300.2576190476193) += (0, 1306.4835352495072) -= (0, 1306.4835352495072)
(15, 2133.0355555555557) += (0, 899.6471868439448) -= (0, 899.6471868439448)
(16, 2662.9955555555553) += (0, 1394.111747916186) -= (0, 1394.111747916186)
};

\end{semilogyaxis}
\end{tikzpicture}

%% file: figs/grid_6x6_runtimes.tex
\begin{tikzpicture}[scale=\fscale]
\begin{semilogyaxis}[
    title={Grid world [$6\times 6$]},
    xlabel={Game size [\# of timesteps]},
    ylabel style={align=center},
    ylabel={Runtime [s]},
    legend pos=north west,
    ymajorgrids=true,
    grid style=dashed,
    scaled ticks=false, 
    title style={font=\titfont},
    label style={font=\labfont},
    tick label style={/pgf/number format/fixed, font=\tickfont}  
]

\legend{rb = 1, rb = 2, rb = 3, rb = 4, rb = 5}
\addplot[color=black,mark=o,error bars/.cd,y dir=both, y explicit]
coordinates {
(6, 8.045333333333334) += (0, 0.15105584971063113) -= (0, 0.15105584971063113) 
(7, 9.322333333333336) += (0, 0.3851271422644893) -= (0, 0.3851271422644893) 
(8, 12.698333333333332) += (0, 0.6468224868554864) -= (0, 0.6468224868554864) 
(9, 16.688333333333333) += (0, 1.1998755762442341) -= (0, 1.1998755762442341) 
(10, 48.902333333333345) += (0, 8.40332054898789) -= (0, 8.40332054898789) 
(11, 53.69066666666668) += (0, 9.092639273070914) -= (0, 9.092639273070914) 
(12, 114.70933333333332) += (0, 42.460032021832234) -= (0, 42.460032021832234) 
(13, 110.65333333333334) += (0, 36.33671815510483) -= (0, 36.33671815510483) 
(14, 987.0109999999999) += (0, 780.3950868460242) -= (0, 780.3950868460242) 
(15, 646.1106666666667) += (0, 411.95404601599694) -= (0, 411.95404601599694) 
(16, 504.06931034482767) += (0, 112.18962921733855) -= (0, 112.18962921733855) 
};

\addplot[color=black,mark=square,error bars/.cd,y dir=both, y explicit]
coordinates {
(6, 8.993) += (0, 0.5335117122673473) -= (0, 0.5335117122673473) 
(7, 8.042666666666669) += (0, 0.07172574416120703) -= (0, 0.07172574416120703) 
(8, 20.966333333333335) += (0, 3.5451691425152267) -= (0, 3.5451691425152267) 
(9, 25.201666666666664) += (0, 3.6609114838685057) -= (0, 3.6609114838685057) 
(10, 86.988) += (0, 11.098350508071006) -= (0, 11.098350508071006) 
(11, 99.9983333333333) += (0, 12.488161558036348) -= (0, 12.488161558036348) 
(12, 910.7836666666668) += (0, 336.64159173122016) -= (0, 336.64159173122016) 
(13, 1083.013) += (0, 429.3825680337976) -= (0, 429.3825680337976) 
(14, 2058.920357142857) += (0, 741.311867794474) -= (0, 741.311867794474) 
(15, 5193.379285714285) += (0, 1725.444084810608) -= (0, 1725.444084810608) 
(16, 4538.722000000001) += (0, 1537.654348629973) -= (0, 1537.654348629973) 
};

\addplot[color=black,mark=pentagon,error bars/.cd,y dir=both, y explicit]
coordinates {
(6, 8.897) += (0, 0.37852926214343124) -= (0, 0.37852926214343124) 
(7, 9.353666666666667) += (0, 0.23188110475693144) -= (0, 0.23188110475693144) 
(8, 29.021) += (0, 5.222093523685204) -= (0, 5.222093523685204) 
(9, 40.955333333333336) += (0, 7.144265874197715) -= (0, 7.144265874197715) 
(10, 164.50900000000001) += (0, 27.795092638685404) -= (0, 27.795092638685404) 
(11, 199.70299999999997) += (0, 30.03790158331197) -= (0, 30.03790158331197) 
(12, 4655.172142857143) += (0, 1455.3678562887847) -= (0, 1455.3678562887847) 
(13, 5188.368846153846) += (0, 2011.4623462868865) -= (0, 2011.4623462868865) 
(14, 3492.7854545454543) += (0, 1800.9295270502987) -= (0, 1800.9295270502987) 
(15, 4104.534210526316) += (0, 1688.3233697295784) -= (0, 1688.3233697295784) 
(16, 4000.960555555556) += (0, 1936.2775728272131) -= (0, 1936.2775728272131) 
};

\addplot[color=black,mark=diamond,error bars/.cd,y dir=both, y explicit]
coordinates {
(6, 95.714) += (0, 6.6168672936374655) -= (0, 6.6168672936374655) 
(7, 55.67080000000001) += (0, 8.954858121340246) -= (0, 8.954858121340246) 
(8, 39.43285714285714) += (0, 4.939541831706727) -= (0, 4.939541831706727) 
(9, 54.00481481481481) += (0, 7.706557099492668) -= (0, 7.706557099492668) 
(10, 370.9034482758622) += (0, 91.00931025733699) -= (0, 91.00931025733699) 
(11, 413.1316666666667) += (0, 107.83054335293538) -= (0, 107.83054335293538) 
(12, 2872.4918181818184) += (0, 1261.968488392274) -= (0, 1261.968488392274) 
(13, 1941.4609523809524) += (0, 1056.045281835311) -= (0, 1056.045281835311) 
(14, 4790.853) += (0, 2089.7571448326235) -= (0, 2089.7571448326235) 
(15, 4843.557894736842) += (0, 1917.4970324129886) -= (0, 1917.4970324129886) 
(16, 4043.782105263158) += (0, 1640.5732600547803) -= (0, 1640.5732600547803) 
};

\addplot[color=black,mark=triangle,error bars/.cd,y dir=both, y explicit]
coordinates {
(6, 11.344000000000001) += (0, 0.07641477861516448) -= (0, 0.07641477861516448) 
(7, 11.73433333333333) += (0, 0.14068688910585636) -= (0, 0.14068688910585636) 
(8, 52.75166666666667) += (0, 8.723207151481777) -= (0, 8.723207151481777) 
(9, 76.03666666666668) += (0, 16.712909301435452) -= (0, 16.712909301435452) 
(10, 569.2244827586208) += (0, 207.58354717327896) -= (0, 207.58354717327896) 
(11, 1931.7936666666667) += (0, 1430.1458703286326) -= (0, 1430.1458703286326) 
(12, 1984.9778260869564) += (0, 674.3612977184093) -= (0, 674.3612977184093) 
(13, 4100.388999999999) += (0, 1699.3507548419934) -= (0, 1699.3507548419934) 
(14, 3124.472142857143) += (0, 1652.1357393396863) -= (0, 1652.1357393396863) 
};

\end{semilogyaxis}
\end{tikzpicture}

%% file: figs/grid_7x7_runtimes.tex
\begin{tikzpicture}[scale=\fscale]
\begin{semilogyaxis}[
    title={Grid world [$7\times 7$]},
    xlabel={Game size [\# of timesteps]},
    ylabel style={align=center},
    ylabel={Runtime [s]},
    legend pos=north west,
    ymajorgrids=true,
    grid style=dashed,
    scaled ticks=false, 
    title style={font=\titfont},
    label style={font=\labfont},
    tick label style={/pgf/number format/fixed, font=\tickfont}  
]

\legend{rb = 1, rb = 2, rb = 3, rb = 4, rb = 5}
\addplot[color=black,mark=o,error bars/.cd,y dir=both, y explicit]
coordinates {
(6, 14.903333333333334) += (0, 2.4122598901813586) -= (0, 2.4122598901813586) 
(7, 13.882857142857143) += (0, 2.0991684456625452) -= (0, 2.0991684456625452) 
(8, 11.48) += (0, 0.1305893235627301) -= (0, 0.1305893235627301) 
(9, 20.707142857142856) += (0, 3.063803047789831) -= (0, 3.063803047789831) 
(10, 36.0775) += (0, 4.329965667135413) -= (0, 4.329965667135413) 
(11, 110.42150000000001) += (0, 12.426868964063312) -= (0, 12.426868964063312) 
(12, 68.85600000000001) += (0, 12.708169407283274) -= (0, 12.708169407283274) 
(13, 113.2435) += (0, 25.11924630401117) -= (0, 25.11924630401117) 
(14, 120.431) += (0, 21.36064397437493) -= (0, 21.36064397437493) 
};

\addplot[color=black,mark=square,error bars/.cd,y dir=both, y explicit]
coordinates {
(6, 17.435000000000002) += (0, 4.116335922707406) -= (0, 4.116335922707406) 
(7, 12.337499999999999) += (0, 0.3098545165534486) -= (0, 0.3098545165534486) 
(8, 11.553333333333333) += (0, 0.1934769581457526) -= (0, 0.1934769581457526) 
(9, 35.5425) += (0, 6.8056113271119365) -= (0, 6.8056113271119365) 
(10, 48.54444444444445) += (0, 10.709981729301658) -= (0, 10.709981729301658) 
(11, 215.80849999999995) += (0, 73.16160824278708) -= (0, 73.16160824278708) 
(12, 214.06650000000005) += (0, 64.81538798478661) -= (0, 64.81538798478661) 
(13, 846.0245000000001) += (0, 252.5552859990055) -= (0, 252.5552859990055) 
(14, 1037.1874999999998) += (0, 446.26545138744433) -= (0, 446.26545138744433) 
};

\addplot[color=black,mark=pentagon,error bars/.cd,y dir=both, y explicit]
coordinates {
(6, 11.01) += (0, 0.13735598518691014) -= (0, 0.13735598518691014) 
(7, 12.831249999999999) += (0, 0.2501459395461549) -= (0, 0.2501459395461549) 
(8, 20.237777777777772) += (0, 2.6720304273975377) -= (0, 2.6720304273975377) 
(9, 49.31222222222223) += (0, 9.476220938306838) -= (0, 9.476220938306838) 
(10, 46.17111111111111) += (0, 5.831286651532891) -= (0, 5.831286651532891) 
(11, 300.9745) += (0, 70.65149057344792) -= (0, 70.65149057344792) 
(12, 354.9085) += (0, 85.46175119307374) -= (0, 85.46175119307374) 
(13, 2266.702857142857) += (0, 1142.5104435675457) -= (0, 1142.5104435675457) 
(14, 3153.435999999999) += (0, 1922.44342560062) -= (0, 1922.44342560062) 
};

\addplot[color=black,mark=diamond,error bars/.cd,y dir=both, y explicit]
coordinates {
(6, 11.744444444444445) += (0, 0.5084975448184775) -= (0, 0.5084975448184775) 
(7, 11.805555555555554) += (0, 0.21002718812831467) -= (0, 0.21002718812831467) 
(8, 12.01125) += (0, 0.2937955891675123) -= (0, 0.2937955891675123) 
(9, 68.3411111111111) += (0, 12.912233265443843) -= (0, 12.912233265443843) 
(10, 67.1888888888889) += (0, 11.390646314269885) -= (0, 11.390646314269885) 
(11, 504.58684210526303) += (0, 155.89034310790316) -= (0, 155.89034310790316) 
(12, 562.6035) += (0, 173.30413361801874) -= (0, 173.30413361801874) 
(13, 6172.482666666667) += (0, 2397.0590012013568) -= (0, 2397.0590012013568) 
(14, 4748.6537499999995) += (0, 1851.936605306473) -= (0, 1851.936605306473) 
};

\addplot[color=black,mark=triangle,error bars/.cd,y dir=both, y explicit]
coordinates {
(6, 25.63888888888889) += (0, 6.403791941855069) -= (0, 6.403791941855069) 
(7, 47.31444444444444) += (0, 1.424462970065006) -= (0, 1.424462970065006) 
(8, 48.596666666666664) += (0, 1.578109804657317) -= (0, 1.578109804657317) 
(9, 73.20666666666666) += (0, 12.825485479223692) -= (0, 12.825485479223692) 
(10, 80.54555555555557) += (0, 16.36560970035023) -= (0, 16.36560970035023) 
(11, 639.3394999999998) += (0, 191.0371810758521) -= (0, 191.0371810758521) 
(12, 1368.9334999999996) += (0, 669.1839942832478) -= (0, 669.1839942832478) 
(13, 1041.793076923077) += (0, 353.78570148446096) -= (0, 353.78570148446096) 
(14, 3779.7430769230773) += (0, 2613.058346720294) -= (0, 2613.058346720294) 
};

\end{semilogyaxis}
\end{tikzpicture}

%% file: sections/conclusion.tex
\section{Conclusion}

In this paper, we introduced Contested Logistics games, a complex logistics problem that incorporates adversarial disruptions. Our model, formulated as a large two-player zero-sum one-shot game on a graph, identifies optimal logistics plans via a (randomized) Nash equilibrium. We demonstrated the computational complexity of finding these equilibria and proposed a practical double-oracle solver using best-response mixed-integer linear programs. Our experiments, conducted on both synthetic and real-world maps, confirm the scalability of our method for reasonably large games. Additionally, our ablation studies underscore the critical importance of explicitly modeling adversarial capabilities, rather than relying solely on heuristic-based logistics plans.
